\documentclass[a4paper,UKenglish,thm-restate]{lipics-v2021}
\usepackage{amsmath}
\usepackage{amsthm}
\usepackage{amsfonts}
\usepackage{graphicx} 
\usepackage[noend,noline,ruled]{algorithm2e}
\usepackage[alignedforall]{lpform}
\usepackage{hyperref}

\usepackage{todonotes}

\newcommand{\Obj}{\mathrm{Obj}}
\newcommand{\J}{\mathcal{J}}

\newcommand{\budgetinstance}{\ensuremath{\mathcal{I}^B}}
\newcommand{\flowenergyinstance}{\ensuremath{\mathcal{I}^{FE}}}
\newcommand{\subsetsuminstance}{\ensuremath{\mathcal{I}^S}}
\newcommand{\eps}{\varepsilon}
\newcommand{\BS}{\text{BS}}
\newcommand{\idling}{\emptyset}
\newcommand{\K}{\mathcal{K}}

\title{Refining the Complexity Landscape of Speed Scaling: Hardness and Algorithms}

\titlerunning{Refining the Complexity Landscape of Speed Scaling: Hardness and Algorithms}

\author{Antonios Antoniadis}{University of Twente, the Netherlands}{a.antoniadis@utwente.nl}{https://orcid.org/0000-0003-2152-7883}{}
\author{Denise Graafsma}{University of Twente, the Netherlands}{d.f.graafsma@utwente.nl}{https://orcid.org/0009-0009-1510-7509}{}
\author{Ruben Hoeksma}{University of Twente, the Netherlands}{r.p.hoeksma@utwente.nl}{https://orcid.org/0000-0002-6553-7242}{}
\author{Maria Vlasiou}{University of Twente, the Netherlands}{m.vlasiou@utwente.nl}{https://orcid.org/0000-0002-0457-2925}{}

\authorrunning{Antonios Antoniadis, Denise Graafsma, Ruben Hoeksma, and Maria Vlasiou}

\Copyright{Antonios Antoniadis, Denise Graafsma, Ruben Hoeksma, and Maria Vlasiou}

\funding{This research is conducted in the scope of Modular Integrated Sustainable Data Center (MISD) project which received funding from the Netherlands Enterprise Agency (RVO) under the 8ra initiative (IPCEI-CIS).}



\nolinenumbers

\begin{document}
\maketitle
\begin{abstract}
    We study the computational complexity of scheduling jobs on a single speed-scalable processor with the objective of capturing the trade-off between the (weighted) flow time and the energy consumption. This trade-off has been extensively explored in the literature through a number of problem formulations that differ in the specific job characteristics and the precise objective function. 
    Nevertheless, the computational complexity of four important problem variants has remained unresolved and was explicitly identified as an open question in prior work (see~\cite{ComplexitySpeedScaling}). In this paper, we settle the complexity of these variants. 
    
    More specifically, we prove that the problem of minimizing the objective of total (weighted) flow time plus energy is NP-hard for the cases of (i) unit-weight jobs with arbitrary sizes, and (ii)~arbitrary-weight jobs with unit sizes. These results extend to the objective of minimizing the total (weighted) flow time subject to an energy budget and hold even when the schedule is required to adhere to a given priority ordering. 
    
    In contrast, we show that when a completion-time ordering is provided, the same problem variants become polynomial-time solvable. The latter result highlights the subtle differences between priority and completion orderings for the problem. 

\end{abstract}

\newpage
\setcounter{page}{1}

\section{Introduction}

Energy is a fundamental and scarce resource that powers every aspect of society. Data centers are among the largest energy consumers globally~\cite{iea}. As the demand for digital services continues to grow, so does the need for more sustainable computing, which in turn calls for the integration of energy efficiency considerations into the design and analysis of algorithms, alongside classical computational resources such as time and space. Within this context,  energy management techniques play a central role in optimizing the trade-off between energy consumption and quality of service. One of the most extensively studied such techniques 
in the algorithmic literature is \emph{speed scaling}. Speed scaling grants the operating system the ability to dynamically adapt the processor speed, allowing higher speeds when performance is critical at the expense of a higher energy consumption, and lower speeds when energy savings are prioritized, for example during off-peak times.

To formalize this trade-off, consider the following underlying scheduling problem: Given a set of~$n$ preemptable jobs, each job~$j$ is associated with a \emph{release time}~$r_j$ denoting the earliest time the job can begin being processed, a \emph{processing volume}~$v_j$ denoting the number of CPU cycles required for processing the job, and a \emph{weight}~$w_j$ denoting the relative importance of the job.
The jobs are to be processed on a speed-scalable processor and the power-consumption corresponding to each allowable speed is given by a power function~$P$. 
The set of different allowable speeds can be discrete or continuous (in the latter case,~$P$ is usually considered to be a continuous convex function of the speed). When the processor runs at a \emph{speed}~$s$, it can process~$s$ units of volume per unit of time. One can think of the speed as the CPU frequency. Naturally, the energy consumption of the processor is defined by the integral over time of~$P(s(t))$, where~$s(t)$ denotes the speed of the processor at time~$t$. A schedule
determines which job is processed at what speed at each time~$t$, so that each job~$j$ is fully processed after its respective release time~$r_j$. As a QoS objective, we consider the total (weighted) \emph{flow time} (also known as \emph{response time}) of a schedule, that is, the (weighted) sum over all jobs of the difference between the time their processing is completed and the time they were released. In other words, the total (weighted) flow time of a schedule is the total accumulated time that jobs spend in the system. 
We are interested in the tradeoff between the energy consumption and the total (weighted) flow time, which we capture by considering as an objective function the sum of energy and total (weighted) flow~time.\footnote{Note that the desired trade-off between energy and flow can be represented by scaling the power function~$P$ accordingly.}

We adopt the quintuple notation~$\star$-$\star\star\star\star$ from~\cite{ComplexitySpeedScaling} to describe different variants of the problem related to this work: The first entry denotes whether the objective is the total (weighted) \emph{F}low plus \emph{E}nergy or the total (weighted) flow time subject to an energy \emph{B}udget (think, e.g., of a battery-powered device). The second entry distinguishes whether the flow is \emph{I}ntegral or \emph{F}ractional\footnote{In this work we only consider the, more common, integral flow -- but keep the tuple entry nevertheless for reasons of consistency.}, the third whether the speed is \emph{C}ontinuous or \emph{D}iscrete, the fourth whether jobs are \emph{W}eighted or \emph{U}nweighted (i.e., all weights are one), and the fifth and final entry whether the sizes are \emph{A}rbitrary or \emph{U}nit.

The combination of speed scaling with a flow time objective  was first explored by Pruhs et al.~\cite{PruhsUW08} who presented a polynomial-time algorithm for the B-ICUU variant of the problem.  The combined objective, which is also the focus of this work, was introduced  by Albers and Fujiwara~\cite{AlbersF07}. They studied the unit-size job setting and proposed a polynomial-time algorithm based on dynamic programming for FE-ICUU, which can also be extended to the energy-budget variant. With respect to the fractional flow objective, Antoniadis et al.~\cite{AntoniadisBCKNP17} showed that FE-FDWA is solvable in polynomial time by an incremental algorithm. On the other hand, the variants~$\star$-I$\star$WA are known to be $\mathsf{NP}$-hard already in the fixed-speed setting~\cite{LABETOULLE1984245}, while Megow and Verschae~\cite{MegowV18} show the $\mathsf{NP}$-hardness of B-IDWA. Finally, Barcelo et al.~\cite{ComplexitySpeedScaling}  settle the computational complexity of variants B-IDUA, B-IDWU, FE-IDUU, and FE-FCWA: The first two are $\mathsf{NP}$-hard whereas the latter two are solvable in polynomial time.
Four variants remained open: FE-IDUA, FE-ICUA, FE-IDWU, FE-ICWU (see also~\cite{ComplexitySpeedScaling} and the open problem discussion in~\cite{barcelo}). Our first main contribution is to resolve those cases by proving that all four variants are $\mathsf{NP}$-hard.

Our $\mathsf{NP}$-hardness results further suggest that simultaneously optimizing the speed(s) at which each job is processed and the order at which jobs are processed is inherently difficult. For the unweighted variants, if the average speed at which each job is processed is given, then an optimal schedule is easily constructed by prioritizing the jobs according to the shortest remaining processing time rule (SRPT). 
However, we show that determining the optimal speeds is $\mathsf{NP}$-hard even when given a priority ordering -- for most variants of the problem even if this priority ordering corresponds to an optimal schedule. In contrast,  as our second main contribution, we also show that a completion ordering is sufficient for efficiently computing optimal speed schedules. We note that such a result for FE-IDUA was claimed before by Barcelo et al.~\cite{ComplexitySpeedScaling} as an extension of an algorithm for FE-IDUU, which was only analyzed for the case of two speeds. They claim that, for this~$2$-speed algorithm, ``it is straightforward to generalize it to~$k$ speeds'', but we provide counterexamples for the two, in our opinion, most natural generalizations (see Section~\ref{sec:natural-generalizations} for more details). Our algorithm on the other hand uses a different, LP-based approach. Furthermore, our results can be extended to the budget version of each variant as well.

\paragraph*{Our contribution}

 Our first main result is showing in Section~\ref{sec:hardness} that FE-IDUA and FE-IDWU are $\mathsf{NP}$-hard.
 \begin{restatable}{theorem}{UAhard}
 \label{thm:UA}
    FE-IDUA is $\mathsf{NP}$-hard.
 \end{restatable}
 \begin{restatable}{theorem}{WUhard}
    FE-IDWU is $\mathsf{NP}$-hard.
    \label{thm:WU}
 \end{restatable}
 The proof of Theorem~\ref{thm:UA} consists of a reduction from the budget variant of the problem B-IDUA with two available speeds, which is known to be $\mathsf{NP}$-hard~\cite{ComplexitySpeedScaling}.
 Via an intricate construction, we are able to simulate an energy-budget constraint within the flow-plus-energy setting for any such two-speed B-IDUA instance. At a high level, the idea is to create an additional job of comparatively large volume along with a substantial number of smaller jobs which are released significantly after all other jobs. In an optimal schedule, these latter small jobs must be processed immediately upon release so as to keep the total flow time low. 
 This forces the large-volume job to complete before the release of the smaller jobs, and in turn,  a specific amount of volume from the original B-IDUA instance needs to be run at speed~$s_2$ thereby simulating the original budget constraint. Our reduction gives further insight into the relationship between the corresponding budget and flow-energy trade-off variants of the~problem.

Note that the above reduction crucially depends on a job of comparatively large volume, something that the FE-IDWU variant does not allow for. As a consequence, the proof of Theorem~\ref{thm:WU} requires a more technically involved argument. The reduction in this case is directly from \textsc{SubsetSum}: instead of introducing a large-volume job, we introduce a large number of low-weight unit-size jobs. However, this causes the challenge that the completion-time ordering of an optimal schedule becomes difficult to predict: small changes in speed might affect the order significantly. We overcome this challenge by exploiting the gap between YES and NO instances in the original proof of the NP hardness of B-IDWU in~\cite{ComplexitySpeedScaling}: By carefully setting the weight of these additional jobs, we are able to simulate an energy-budget constraint while at the same time not affecting the total (weighted) flow of the original jobs by too much. 

As a corollary, both $\mathsf{NP}$-hardness results can be extended to the respective continuous-speed problems. Therefore, Theorems~\ref{thm:UA} and~\ref{thm:WU} resolve the complexity of the last four unresolved variants and complete the computational complexity landscape of flow-plus-energy and flow under an energy-budget~problems.  

For each of the considered problem variants, any algorithm must address two decisions for each interval of time: determine which job to run and select a speed at which to run it. We dive deeper into the complexity of the problem by exploring how the availability of different types of information regarding which job to run affects the computational complexity of the problem. We note that given an optimal speed profile (or equivalently, the average speed of each job under an optimal schedule), two types of such information are sufficient to efficiently compute an optimal schedule: a \emph{priority ordering} (specifying which job among the available ones should be processed) and a \emph{completion-time ordering} (specifying the exact order in which the jobs should complete).

We show that our $\mathsf{NP}$-hardness reductions can be extended to show that all considered problem variants remain $\mathsf{NP}$-hard, even if schedules are required to adhere to a given priority ordering. We discuss in more detail and prove these extensions in Section~\ref{sec:hardness-priority}.

At first glance, being given a completion-time ordering might seem as informative as being given a priority ordering. However, somewhat surprisingly, we show that all considered variants become solvable in polynomial time when a completion-time ordering is given.

We extend the quintuple notation by the suffix -P or -C, indicating that feasible schedules must adhere to a given priority ordering or completion-time ordering, respectively. In Section~\ref{sec:lp}, we obtain the following result regarding completion-time orderings.

 \begin{restatable}{theorem}{LPresult}
    There is a polynomial-time algorithm that computes optimal schedules for all~$\star$-ID$\star\star$-C variants. 
    \label{thm:lp-intro}
 \end{restatable}

In addition to providing efficient algorithms, Theorem~\ref{thm:lp-intro} allows for a better understanding of the interplay between deciding on the job order and on the processor speed at each time. As a consequence, our results provide more insight into where the complexity of the problem stems from.
 In particular, when combined with our previous $\mathsf{NP}$-hardness results, Theorem~\ref{thm:lp-intro} directly implies that for~$\star$-IDUA,~$\star$-IDWU and~$\star$-IDWA it is $\mathsf{NP}$-hard to compute an optimal completion-time ordering.

 The algorithm implied by Theorem~\ref{thm:lp-intro} is based on a novel LP formulation for FE-ID-$\star\star$-C that effectively makes use of the structure imposed by the completion-time ordering.  The LP formulation can be adapted to solve the respective budget variants B-ID$\star\star$-C as well, with a simple modification in the objective function and the addition of one constraint.

\paragraph*{Further related work}

Speed scaling was introduced to the algorithmic literature thirty years ago by Yao, Demers and Shenker~\cite{YaoDS95}. Their seminal work focused on minimizing the total energy consumption of a schedule subject to deadline constraints and considered both the offline and the online settings. Since then, a plethora of different speed scaling settings have been considered, one of which is that of flow times as a QoS -- motivated by the fact that, in many practical situations, jobs are not labeled with deadlines. 

In addition to the results already discussed, the problem has been extended to multiprocessors (see~\cite{Bunde06}, which also shows that exact optimal B-ICWU schedules cannot be computed even on a machine 
that can do real arithmetic operations and root computations). The online setting of different variations has also been widely explored; see~\cite{AlbersF07,BansalPS07,GuptaIKMP12,LamLTW13, BansalC09} for some examples.

For a more extensive overview of the literature on energy-efficient scheduling, we refer the reader to the surveys by Albers~\cite{Albers10-survey} and Irani and Pruhs~\cite{IraniP05}.

\section{Formal Problem Description and Notation}
Consider~$n$ \emph{jobs}~$1,\ldots,n$ that have to be scheduled on a single speed-scalable processor. Each job~$j$ is associated with a corresponding \emph{release time}~$r_j\in\mathbb{Q}_{\geq0}$, a \emph{processing volume}~$v_j\in\mathbb{Q}_{\geq0}$, and \emph{weight}~$w_j\in\mathbb{Q}_{\geq0}$. The processor is equipped with~$k+1$ distinct nonnegative rational \emph{speeds}:~$s_0< \dots < s_k$, with corresponding nonnegative rational \emph{power consumptions}~$P_0<\dots < P_k$. If the processor runs at speed~$s_i$ for~$\tau$ time units, it processes a volume of~$\tau\cdot s_i$ and consumes an amount~$\tau\cdot P_i$ of energy. We assume that the processor idles when it uses the lowest speed (that is~$s_0=0$), that~$P_0=0$, and w.l.o.g., that $\min_j\{r_j\}=0$.

A \emph{schedule} is defined as~$\sigma=(J,S)$. 
Here,~$J\left(t\right)\in \{\idling, 1,\dots, n\}$ specifies which job is being processed at time~$t$, where~$\idling$ indicates the processor is idle. Any job~$j$ can only be processed no later than its release time~$r_j$. In other words for all~$j$, and~$t\in [0,r_j)$,~$J(t)\neq j$. Preemption of jobs is allowed. Similarly,~$S\left(t\right)\in \{0, 1,\dots,k\}$ specifies at what speed the server runs at time~$t$. For~$i\in \{1,\dots, k\}$,~$S(t)=i$ indicates that the processor runs at speed~$s_i$. 
We define
\begin{align*}
    X_j(\sigma) := \{t: J(t) = j \wedge S(t)>0\}\,.
\end{align*}
Note that~$X_j(\sigma)$ is the union of disjoint intervals; it consists of exactly the times during which job~$j$ is processed under schedule~$\sigma$. We use~$x_j(\sigma)$
to denote the total duration during which~$j$ runs under~$\sigma$ and~$\chi(\sigma) = \sum_{j=1}^n x_j(\sigma)$ to denote the total amount of time during which~$\sigma$ is executing some job. We refer to~$x_j(\sigma)$ as the \emph{processing time} of~$j$ in~$\sigma$. We say that a schedule is \emph{feasible} if every job is fully processed no earlier than its release time. That is if it satisfies the following for all~$j\in\{1,\dots,n\}$:
\begin{enumerate}
    \item~$\int_{t\in X_j(\sigma)}{s_{S(t)}dt}=v_j$, and 
    \item we have~$t\geq r_j$, for all~$t\in X_j(\sigma)$.       
\end{enumerate}

The \emph{completion time} of~$j$ in~$\sigma$, denoted by~$C_j(\sigma)$, is the maximum~$t$ in~$X_j(\sigma)$. The \emph{flow time} of~$j$ in~$\sigma$ is defined as~$F_j\left(\sigma\right):=C_j\left(\sigma\right)-r_j$ 
and the total (weighted) flow time of~$\sigma$ is 
\[
    F\left(\sigma\right):=\sum_{j=1}^{n}{w_j F_j\left(\sigma\right)}\,.
\] 
The total energy consumption of~$\sigma$ is naturally defined as the power of~$\sigma$ integrated over time: $E\left(\sigma\right):=\int_{t=0}^{\infty}P_{S\left(t\right)}dt$.

For a schedule~$\sigma$, we say that~$j$ runs at speed~$s_j(\sigma) := \frac{v_j}{x_j(\sigma)}$ (with slight abuse of notation). Note that~$s_j(\sigma)$ refers to the average processing speed of~$j$ under~$\sigma$. To distinguish this from the speeds~$s_1,\dots, s_k$ that the processor is allowed to run at as per the problem statement, we will often refer to the latter as the given or allowed speeds. 
For any~$i\in \{1,\dots k\}$, we denote by~$x_j^i = v_j/s_i$ and~$E_j^i =x_j^i P_i$, the processing time and energy consumption of~$j$ respectively under the assumption that~$j$ is processed solely at the given speed~$s_i$.
It is an easy observation that for any~$j$ and any feasible schedule~$\sigma$,~$x_j(\sigma)$ and~$E_j(\sigma)$ can be expressed as a convex combination of the~$x_j^i$'s and~$E_j^i$'s:
\begin{observation}\label{obs:proc_times_conv_comb}
    For any job~$j$ and~$\lambda_j^1,\dots, \lambda_j^k\geq 0$ with~$\sum_{i=1}^k \lambda_j^i=1$, there is a schedule~$\sigma$ with~$x_j(\sigma) = \sum_{i=1}^k \lambda_j^i x_j^i$ and~$E_j(\sigma) =  \sum_{i=1}^k \lambda_j^i E_j^i$. 
\end{observation}
\begin{proof}
    Processing an amount~$v_j$ of volume at a speed of~$s_i$, requires~$x_j^i$ time and~$E_j^i$ energy. Processing job~$j$ at speed~$s_i$ for~$\lambda_j^i x_j^i$ time results in~$\lambda_j^i v_j$ volume being processed at an energy of~$\lambda_j^i E_j^i$. Let~$y_j^i(\sigma)$ be the amount of time~$\sigma$ processes job~$j$ at speed~$s_i$. To be more precise:~$ y_j^i(\sigma) = \left|\{t: J(t) = j, S(t) = i\}\right|$. The result is obtained by setting~$\lambda_j^i =  y_j^i(\sigma)/x_j^i$.
\end{proof}
In general, one can obtain any convex combination~$x_j(\sigma) = \sum_{i=1}^k \lambda_j^i x_j^i$ with~$E_j(\sigma) =  \sum_{i=1}^k \lambda_j^i E_j^i$ by using each given speed~$s_i$ for~$\lambda_j^ix_j^i$ time during~$X_j(\sigma)$. Suppose that for a processing time of~$x_i$, rather than running at speed~$s_i$, we use a combination of two different given speeds~$s_q$ and~$s_r$. If this results in a lower energy consumption, then we would always prefer to use this combination of~$s_q$ and~$s_r$ and never use~$s_i$. Hence, to avoid any given speeds being superfluous, we have the following.

\begin{observation}\label{obs:no_conv_comb}
    Let~$i\in \{2,\dots,k-1\}$,~$q\in\{1,i-1\}$, and~$r\in\{i+1,k\}$. For a job~$j$, let~$\lambda\in(0,1)$ be such that~$x_j^i= \lambda x_j^q+ (1-\lambda)x_j^r$. Then~$E_j^i<\lambda E_j^q+ (1-\lambda)E_j^r$.
\end{observation}
\begin{proof}
    If~$E_j^i\geq\lambda E_j^q+ (1-\lambda)E_j^r$, then rather than ever using speed~$s_i$, we can use the corresponding combination of~$ s_q$ and~$ s_r$, which gives the same processing time, but with the same or lower energy consumption. Since~$x_j^i= \lambda x_j^q+ (1-\lambda)x_j^r$ implies~$\frac{1}{s_i}= \lambda \frac{1}{s_q}+ (1-\lambda)\frac{1}{s_r}$, this then holds for every job, hence~$s_i$ will be superfluous. 
\end{proof}

The following observation follows from Observation \ref{obs:no_conv_comb} and shows us that while processing a job, we either use one given speed, or two consecutive speeds~$s_i$ and~$s_{i+1}$.

\begin{restatable}{observation}{obsConv}\label{obs:convex_combination_speeds}
Consider a schedule~$\sigma$ and job~$j$ with~$s_i\leq s_j(\sigma)\leq s_{i+1}$. Let~$\lambda\in[0,1]$ be such that~$x_j(\sigma) =  \lambda x_j^{i} + (1-\lambda) x_j^{i+1}$, then~$E_j(\sigma) =\lambda E_j^{i} + (1-\lambda)E_j^{i+1}$.
\end{restatable}
\begin{proof}
    The proof can be found in Appendix \ref{app:problem_desc}.
\end{proof}

\paragraph*{Problem Variants.}
Recall that we adopt the quintuple notation~$\star$-$\star\star\star\star$ from~\cite{ComplexitySpeedScaling}. For completeness, we formally restate the relevant specific entries with the notation introduced above.  For the first entry, we consider two different objectives: either minimize~$F(\sigma) + E(\sigma)$, or minimize only~$F(\sigma)$ subject to an energy constraint~$E(\sigma)\leq B$ for some given budget~$B$. We refer to the former as the \emph{flow + energy} variant (\emph{FE}), and the latter as the \emph{budget} variant (\emph{B}). For the second entry, we only consider \emph{integral flow (I)}~$F(\sigma)$ as defined above, but the related concept of a \emph{fractional flow (F)} has been considered in the literature. For the third entry, the allowed speeds of the processor could either be \emph{discrete (D)} as defined above, or \emph{continuous (C)}, where the set of allowable speeds is~$[0,H)$ for some~$H>0$, or~$[0,\infty)$ and the power consumption is defined by a function~$P$ of the speed with~$P(0)=0$. Note that Observation~\ref{obs:convex_combination_speeds} implies that the discrete speed setting can be modeled by the continuous one with a piecewise linear power function~$P$ and the discrete speed setting is therefore a special case of the continuous setting (see also Theorem~7 in~\cite{ComplexitySpeedScaling}). For the fourth entry, in the \emph{unweighted} variant (\emph{U}), we have~$w_1=\dots=w_n=1$, whereas for the \emph{weighted} variant (\emph{W}) we only require that~$w_j\in\mathbb{Q}_{>0}$ for all~$j\in \{1,\dots,n\}$. Finally, for the fifth entry, in the \emph{unit size} variant (\emph{U}), we have~$v_1=\dots=v_n=1$, whereas in the \emph{arbitrary size} variant (\emph{A}) we only require that~$v_j\in\mathbb{Q}_{>0}$ for all~$j\in\{1,\dots, n\}$. Recall that we also consider each variant with a fixed priority (P) or completion (C) ordering and extend the quintuple notation by a suffix accordingly.

Note that the above definition of a schedule $\sigma=(J,S)$ applies to all described variants of the problem.

\subsection{Affection and Shrinking Energy}
In this section we introduce the notions of \emph{affection} and \emph{shrinking energy}.  Intuitively, the affection of a job measures the improvement in total flow resulting from processing that job in less time (i.e., at higher speed), while its shrinkage energy measures the energy needed to achieve this reduction in the processing time of the job. The affection and shrinking energy combined, provide information on which jobs (if any) one could accelerate in order to obtain an improved schedule. Similar definitions are used in~\cite{ComplexitySpeedScaling}, however, we do deviate from their notation, so as to be able to use it in a more general setting.

\begin{definition}[Affection]
Consider a schedule~$\sigma$.
\begin{itemize}
    \item If~$C_j(\sigma) \leq C_{j'}(\sigma)$ and~$C_j(\sigma)>r_{j'}$, then~$j$ \emph{affects}~$j'$.
    \item If~$j$ affects~$j'$ and~$j'$ affects~$j''$, then~$j$ affects~$j''$.
    \item The \emph{affection set} of~$j$ in~$\sigma$, denoted by~$K_j(\sigma)$, is the set of jobs affected by~$j$ in~$\sigma$. We define the \emph{(weighted) affection} as $\kappa_j(\sigma) := \sum_{j'\in K_j(\sigma)} w_{j'}$.
\end{itemize}  
\end{definition}

The~$\emph{lower affection}$, denoted~$\kappa^+(\sigma)$, is defined analogously to affection, but additionally allows~$C_j(\sigma)= r_{j'}$. Note that a job always (lower) affects itself. Intuitively, the affection of job~$j$ captures how much the flow will improve if we reduce the processing time of~$j$.

\begin{restatable}{observation}{obsAff}\label{obs:affection}
    Let~$\sigma$ be a schedule and consider the processing time~$x_j(\sigma)$ of~$j$ in~$\sigma$. Then there is an~$\varepsilon'>0$ such that for all~$0<\varepsilon\leq\varepsilon'$, the following holds.
    \begin{itemize}
        \item Let~$\sigma^{-\varepsilon}$ be a schedule with the same completion order as~$\sigma$, such that~$x_{j'}(\sigma^{-\varepsilon}) = x_{j'}(\sigma)$ for all~$j'\neq j$, and~$x_j(\sigma^{-\varepsilon}) = x_j(\sigma)-\varepsilon$. Then~$   F(\sigma^{-\varepsilon}) = F(\sigma)-\varepsilon\kappa_j(\sigma)$. 
        \item Let~$\sigma^{+\varepsilon}$ be a schedule with the same completion order as~$\sigma$, such that~$x_{j'}(\sigma^{+\varepsilon}) = x_{j'}(\sigma)$ for all~$j'\neq j$, and~$x_j(\sigma^{+\varepsilon}) = x_j(\sigma)+\varepsilon$. Then~$
        F(\sigma^{+\varepsilon}) = F(\sigma)+\varepsilon\kappa^+_j(\sigma)$.
    \end{itemize}
\end{restatable}
\begin{proof}
    The first statement follows by observing that when~$C_j(\sigma) \leq C_{j'}(\sigma)$ and~$C_j(\sigma)>r_{j'}$, 
    reducing the processing time of~$j$ allows to also reduce the completion time of~$j'$ (without changing $x_{j'}$), and by transitiveness the same holds for any~$j''$ with~$C_{j'}(\sigma) \leq C_{j''}(\sigma)$ and~$C_{j'}(\sigma)>r_{j''}$.
    The second statement can be argued analogously.
\end{proof}

\begin{definition}
    Let~$s_1,\dots, s_k$ be given speeds. We define $\Delta_i := \left(P_{i+1}s_i-P_i s_{i+1}\right)/\left(s_{i+1} - s_i\right)$
    for~$i\in\{1,\dots, k-1\}$, and~$\Delta_0=-\infty$, and~$\Delta_k := \infty$.
\end{definition}

\begin{definition}[Shrinking/expanding energy]
    Let~$\sigma$ be a schedule. For any job~$j$,
    \begin{itemize}
        \item if~$s_j(\sigma)\in[s_i,s_{i+1})$, the \emph{shrinking energy} of~$j$ in~$\sigma$ 
        is~$\Delta_j(\sigma):=\Delta_i$.
        \item if~$s_j(\sigma)\in(s_i,s_{i+1}]$, the \emph{expanding energy} of~$j$ in~$\sigma$ 
        is~$\Delta^+_j(\sigma):=\Delta_i$.
    \end{itemize}
    If~$s_j(\sigma)=s_k$, the \emph{shrinking energy} is~$\Delta_j(\sigma):=\Delta_k=\infty$, 
    while if~$s_j(\sigma)=s_1$ the \emph{expanding energy} of~$j$ in~$\sigma$ 
        is~$\Delta^+_j(\sigma):=\Delta_0=-\infty$.
\end{definition}
Note that if~$s_j(\sigma)\in(s_i,s_{i+1})$, then~$\Delta_j(\sigma)=\Delta_j^+(\sigma)$.

By Observation~\ref{obs:affection}, the (lower) affection of~$j$ quantifies the change in flow for small changes in the processing time of~$j$. The following lemma gives us an analogous result for the shrinking/expanding energy and the change in energy consumption.

\begin{restatable}{lemma}{lemmaDelta}\label{lemma:delta}
    Let~$\sigma$ be a schedule where job~$j$ has processing time~$x_j(\sigma)\in[x_j^{i},x_j^{i+1}]$, and let~$\eps>0$. Then, 
    \begin{itemize}
        \item if~$\eps\leq x_j(\sigma)-x_j^{i+1}$,
        the schedule~$\sigma^{-\varepsilon}$ with~$x_j(\sigma^{-\varepsilon}) = x_j(\sigma)-\varepsilon$ and~$x_{j'}(\sigma^{-\varepsilon}) = x_{j'}(\sigma)$ for all~$j'\neq j$ has total energy consumption equal to
            $E(\sigma^{-\varepsilon}) = E(\sigma)+\varepsilon\Delta_j(\sigma)$.
        \item if~$\eps\leq x_j^{i}-x_j(\sigma)$,
        the schedule~$\sigma^{+\varepsilon}$ with~$x_j(\sigma^{+\varepsilon}) = x_j(\sigma)+\varepsilon$ and~$x_{j'}(\sigma^{+\varepsilon}) = x_{j'}(\sigma)$ for all~$j'\neq j$ has total energy consumption equal to
            $E(\sigma^{+\varepsilon}) = E(\sigma)-\varepsilon\Delta^+_j(\sigma)$.
    \end{itemize}
\end{restatable}
\begin{proof}
    We only prove the first statement, as the second can be proven with an analogous argument. 
    
    Suppose~$\eps\leq x_j(\sigma)-x_j^{i+1}$ and let~$0<\lambda\leq 1$ be such that~$x_j(\sigma) = \lambda x_j^i + (1-\lambda)x_j^{i+1}$. By Observation~\ref{obs:convex_combination_speeds}, we have~$E_j(\sigma) = \lambda E_j^i + (1-\lambda) E_j^{i+1}$.
     Since~$x_j(\sigma^{-\varepsilon})\geq x_j^{i+1}$, there exists a~$\Tilde{\lambda}\in[0,\lambda)$ such that 
    $x_j(\sigma^{-\varepsilon}) = \Tilde{\lambda} x_j^i + (1-\Tilde{\lambda})x_j^{i+1}$,
    which gives us $E_j(\sigma^{-\varepsilon}) = \Tilde{\lambda} E_j^i + (1-\Tilde{\lambda}) E_j^{i+1}$.
    Then
   \[
        \varepsilon 
        = (\lambda-\Tilde{\lambda})\left(x_j^i - x_j^{i+1}\right)
        = (\lambda-\Tilde{\lambda})\left(\frac{v_j}{s_i} - \frac{v_j}{s_{i+1}}\right)  =\frac{v_j(\lambda-\Tilde{\lambda})}{s_i s_{i+1}}(s_{i+1}-s_i)\,,
   \]
    and we obtain that
    \begin{align*}
       E_j(\sigma^{-\varepsilon}) - E_j(\sigma) &= (\Tilde{\lambda}-\lambda)\left(E_j^i - E_j^{i+1}\right)
       = v_j(\lambda-\Tilde{\lambda})\left(\frac{P_i}{s_i} - \frac{P_{i+1}}{s_{i+1}}\right)\\
       &= \frac{v_j(\lambda-\Tilde{\lambda})}{s_i s_{i+1}} (P_{i+1}s_i-P_i s_{i+1})
       = \varepsilon \frac{P_{i+1}s_i-P_i s_{i+1}}{s_{i+1}-s_i} 
       = \varepsilon\Delta_j(\sigma)\,. \tag*{\qedhere}
    \end{align*}
\end{proof}

\begin{corollary}\label{cor:kappa_delta_balance}
    Let~$\sigma^*$ be an optimal schedule for FE-ID**.  Then for any job~$j$
    \begin{itemize}
        \item~$\kappa_j(\sigma^*) - \Delta_j(\sigma^*) \leq 0$, and
        \item~$\Delta_j^+(\sigma^*) - \kappa_j^+(\sigma^*) \leq 0$.
    \end{itemize}
\end{corollary}
\begin{proof}
    The result follows directly from Observation~\ref{obs:affection} and Lemma~\ref{lemma:delta}.
\end{proof}

\section{Hardness Results}
\label{sec:hardness}

\subsection{Hardness of FE-IDUA}\label{sec:hardness_FEUA}
We show that FE-IDUA is $\mathsf{NP}$-hard through a reduction from B-IDUA with two speeds, a problem shown to be $\mathsf{NP}$-hard by Barcelo et al.~\cite[proof of Theorem 1]{ComplexitySpeedScaling}. 

Consider an instance~$\mathcal{I}^B$ of B-IDUA with given speeds~$s_1<s_2$, power consumptions~$P_1<P_2$, budget~$B$, 
and~$n$ jobs with volumes~$v_1,\dots,v_n$ and release times~$r_1,\dots,r_n$. 
We assume w.l.o.g.\ that~$s_1=1$. 
The following notation is used in the construction of an instance~$\flowenergyinstance$, based on~$\budgetinstance$, that establishes the reduction from B-IDUA to FE-IDUA.
Let~$V= \sum_{j\in\{1,\dots, n\}} v_j$ be the total volume of jobs in~$\mathcal{I}^B$, and let~$Y := \frac{B - P_1V}{\Delta_1}$. Intuitively, a schedule that only utilizes~$s_1$ already consumes an amount~$P_1V$ of energy. 
Thus,~$Y$ denotes the maximum decrease of the total processing time that the budget allows, 
compared to such a schedule. 
Note that, w.l.o.g.,
\begin{align}\label{eq:budget_inbetween}
    P_1 V <  B < P_2\frac{V}{s_2}\,,
\end{align}
since, if this is not the case, then~$\mathcal{I}^B$ is either infeasible or trivial. 

Let~$\sigma^*$ be an optimal schedule for~$\mathcal{I}^B$, let~$\sigma_1$ be the SRPT schedule for~$\mathcal{I}^B$ that processes every job at a speed of exactly~$s_1$, let~$C_{\max}(\sigma_1)$ be the makespan of~$\sigma_1$, and let~$I(\sigma_1)=C_{\max}(\sigma_1)-V$ the total idle time (recall that~$\min\{r_j\} = 0$). 

The following observation follows directly from \eqref{eq:budget_inbetween} and Lemma~\ref{lemma:delta}.

\begin{observation}\label{obs:reduction_shrinkage}
    For any optimal schedule~$\sigma^*$ for~$\mathcal{I}^B$ it holds that $\chi(\sigma^*) = V - Y$.
\end{observation}
That is, an optimal schedule for~$\budgetinstance$ uses the complete budget~$B$ to shrink the total processing time to~$V-Y$.

We construct~$\mathcal{I}^{FE}$ as follows.
\begin{itemize}
    \item \emph{Jobs.}~$\mathcal{I}^{FE}$ has~$2n+1$ jobs: Jobs~$1$ to~$n$ are identical to the ones of~$\mathcal{I}^B$. I.e.,~$\tilde{r}_i=r_i$ and~$\tilde{v}_i=v_i$, for~$i=1,\dots, n$.  
    Job~$n+1$ has volume~$\tilde{v}_{n+1}=I(\sigma_1) + (s_2+1)V+Y$ and~$\tilde{r}_{n+1}=0$ and jobs~$n+2,\dots, 2n+1$ all have a volume of~$Y+1$ and a release time of~$C_{max}(\sigma_1)+(s_2+1)V$.
    \item \emph{Processor.} The processor has the same two speeds~$\tilde{s}_1=s_1$ and~$\tilde{s}_2=s_2$ as in~$\mathcal{I}^B$. 
    The corresponding power consumptions are given by~$\tilde{P}_1 = P_1$, and~$\tilde{P}_2=(n+3/2)(s_2-1)+P_1s_2$, and the associated shrinking energy~$\tilde{\Delta}_1=\frac{\tilde{P}_{2}\tilde{s}_1-\tilde{P}_1 \tilde{s}_{2}}{\tilde{s}_{2} - \tilde{s}_1}$.
\end{itemize}

\begin{figure}[t]
\includegraphics[width = 0.85\textwidth]{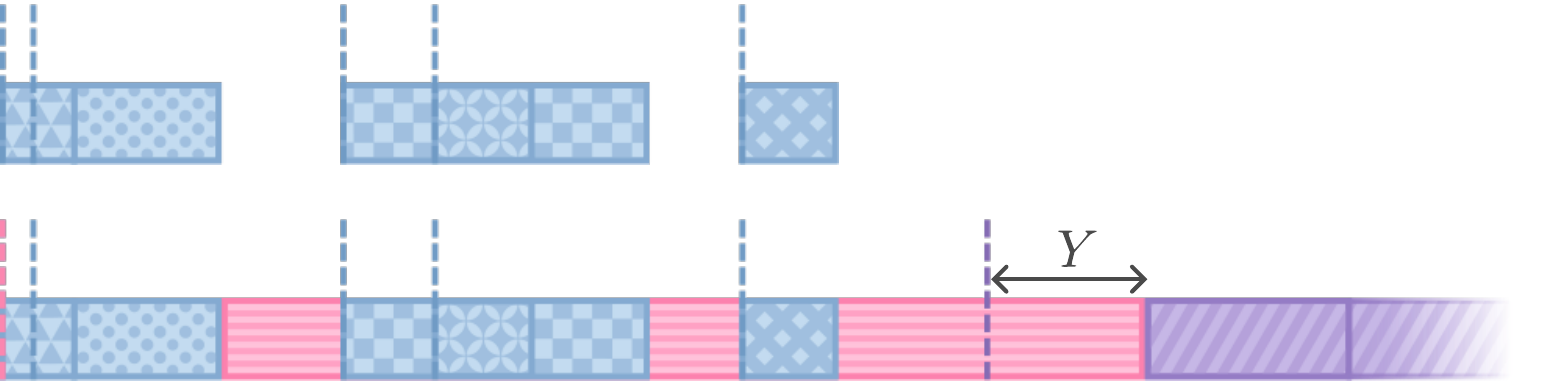}
\centering
\caption{An instance of B-IDUA (depicted at the top) is transformed into an instance of FE-IDUA (depicted at the bottom).}
\label{fig:reduction}
\end{figure} 

Figure \ref{fig:reduction} depicts how \flowenergyinstance\ is constructed from \budgetinstance. Note that \flowenergyinstance\ can be computed in polynomial time and that for the shrinking energy in~\flowenergyinstance, we have $n+1 <\Tilde{\Delta}_1 < n+2$.

The intuition behind the reduction is that in an optimal schedule~$\tilde{\sigma}^*$ for~$\mathcal{I}^{FE}$
exactly~$B$ energy is consumed to process jobs~$1,\dots, n$, and in turn these are processed ``independently'' of the other jobs. Therefore,~$\tilde{\sigma}^*$ produces an optimal schedule for~$\mathcal{I}^B$.
The schedule that~$\tilde{\sigma}^*$ produces for~$\mathcal{I}^B$ is simply~$\tilde{\sigma}^*$ restricted to jobs~$\{1,\dots,n\}$, which we define as~$\tilde{\sigma}^*_{\{1,\dots,n\}}=\left(\tilde{J}^*_{\{1,\dots,n\}},\tilde{S}^*_{\{1,\dots,n\}}\right)$, with
\[
    \left(\tilde{J}^*_{\{1,\dots,n\}}(t),\tilde{S}^*_{\{1,\dots,n\}}(t)\right) =  
        \begin{cases}
            \left(\tilde{J}^*(t),\tilde{S}^*(t)\right), &\text{if~$\tilde{J}^*(t)\in\{1,\dots,n\}$}\\
            (\emptyset,\emptyset), &\text{otherwise}
        \end{cases}\,.
\]

\begin{lemma}\label{obs:reduction_order}
    In an optimal solution for \flowenergyinstance, jobs~$n+1,\dots,2n+1$ complete last.
\end{lemma}
\begin{proof}
    First consider the schedule where job~$n+1$ is processed at speed~$s_2$, while all other jobs 
    are scheduled at speed~$s_1$. From the fixed-speed version of the problem---and ignoring 
    energy cost---we know that the optimal order to process the jobs in that case is by using 
    the SRPT rule~\cite{schrage1968}. By construction, job~$n+1$ has larger processing time than 
    any job~$j$ in~$\{1,\dots,n\}$ and, therefore at time~$r_j$, completes later than each. For any schedule, 
    we can argue the same. Note also that, in turn,~$r_j$ for~$j\in\{n+2,\dots,2n+1\}$ is larger than the 
    last completion time of the jobs in~$\{1,\dots,n\}$ in any schedule. Thus, the lemma holds.
\end{proof}

\begin{figure}[t]
\includegraphics[width = 0.85\textwidth]{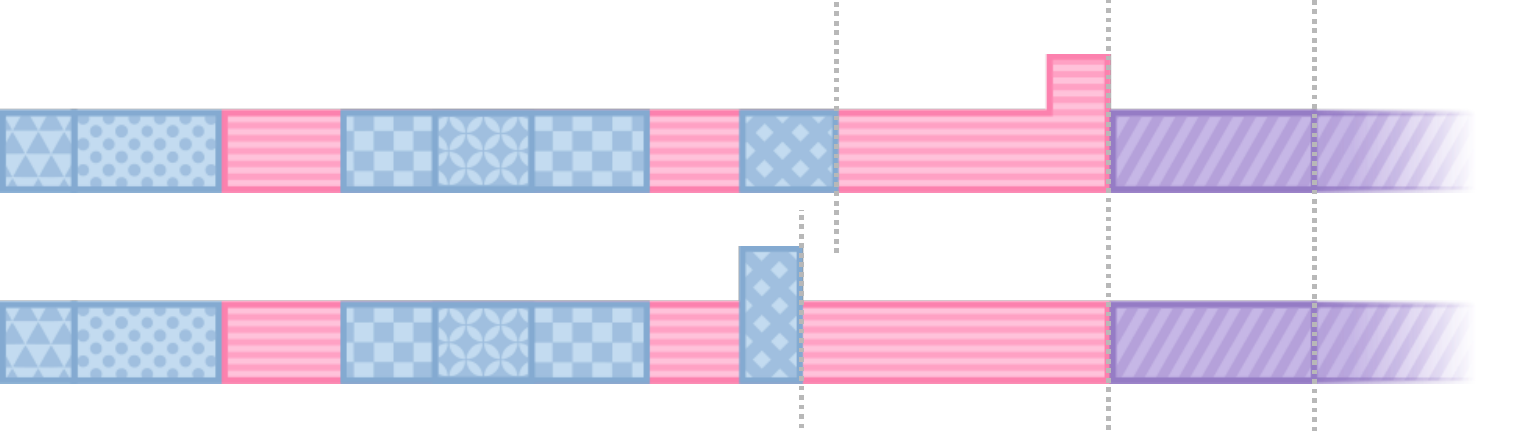}
\centering
\caption{A comparison between speeding up job~$n+1$ and job~$n$. The dashed lines indicate the completion times and show that speeding up job~$n+1$ improves the flow of at most~$n+1$ jobs, while speeding up job~$n$ improves the flow of at least~$n+2$ jobs.}
\label{fig:reduction_old_better}
\end{figure}

\begin{lemma}
    In an optimal solution for \flowenergyinstance, jobs~$n+1,\dots,2n+1$ have average speed~$s_1$.
    \label{lem:slow-jobs}
\end{lemma}
\begin{proof}
    Since~$\tilde\Delta_1 > n+1$, and jobs~$n+1,\dots,2n+1$ complete last in an optimal schedule, by Corollary~\ref{cor:kappa_delta_balance}, these jobs run at speed~$s_1$.
\end{proof} 
Figure~\ref{fig:reduction_old_better} depicts the difference between speeding up a job in~$\{1,\dots,n\}$ and a jobs in~$\{n+1,\dots,2n+1\}$.

\begin{lemma}\label{lemma:reduction_Y_on_old}
    For \flowenergyinstance, any optimal solution uses a total amount of time of~$V-Y$ for jobs~$1,\dots, n$.
    \label{lem:budget}
\end{lemma}
\begin{proof}
Consider an optimal schedule~$\tilde\sigma^*$ for \flowenergyinstance. 
Assume first, for the sake of contradiction, that~$\tilde\sigma^*$ uses strictly less time than~$V-Y$ for jobs~$1,\dots, n$. Then, until~$r_{n+2}$ the amount of time that we do not process a job in~$\{1,\dots, n\}$ is strictly greater than $r_{n+2}-V +Y = v_{n+1}$. Hence, job~$n+1$ must have a completion time strictly smaller than~$r_{n+2}$. In other words, for any~$j\in \{1, \dots, n+1\}$ we have~$\kappa_j^+(\tilde\sigma^*) \le n+1 < \tilde\Delta_1^+(\tilde\sigma^*)$. This contradicts~Corollary~\ref{cor:kappa_delta_balance}.

Assume next, for the sake of contradiction, that~$\tilde\sigma^*$ uses strictly more time than~$V-Y$ for jobs~$1,\dots, n$. Then, by Lemma~\ref{lem:slow-jobs} it must be the case that~$C_{n+1}(\tilde\sigma^*)>r_{n+2}$ and therefore any job~$j\in\{1,\dots, n\}$  must have~$\kappa_j\ge  n+2$. This again contradicts Corollary~\ref{cor:kappa_delta_balance}.
\end{proof}

By Lemma \ref{lem:slow-jobs} and Lemma \ref{lemma:reduction_Y_on_old}, we have that in an optimal schedule for \flowenergyinstance, only jobs in~$\{1,\dots,n\}$ are sped up, and they have a total processing time of~$V-Y$. Note that by Observation \ref{obs:reduction_shrinkage}, this also holds for an optimal schedule for \budgetinstance. Hence, it follows that if we take an optimal schedule for \budgetinstance\ and add the jobs~$n+1,\dots,2n+1$ at speed~$s_1$, then we obtain a schedule for \flowenergyinstance\ that satisfies these conditions, as depicted in Figure \ref{fig:reduction_opt_for_both}. This naturally leads to the following question: If we take an optimal schedule for \flowenergyinstance\ and remove jobs~$n+1,\dots,2n+1$, do we obtain an optimal schedule for \budgetinstance? The following lemma shows that this indeed the case.

\begin{figure}[t]
\includegraphics[width = 0.85\textwidth]{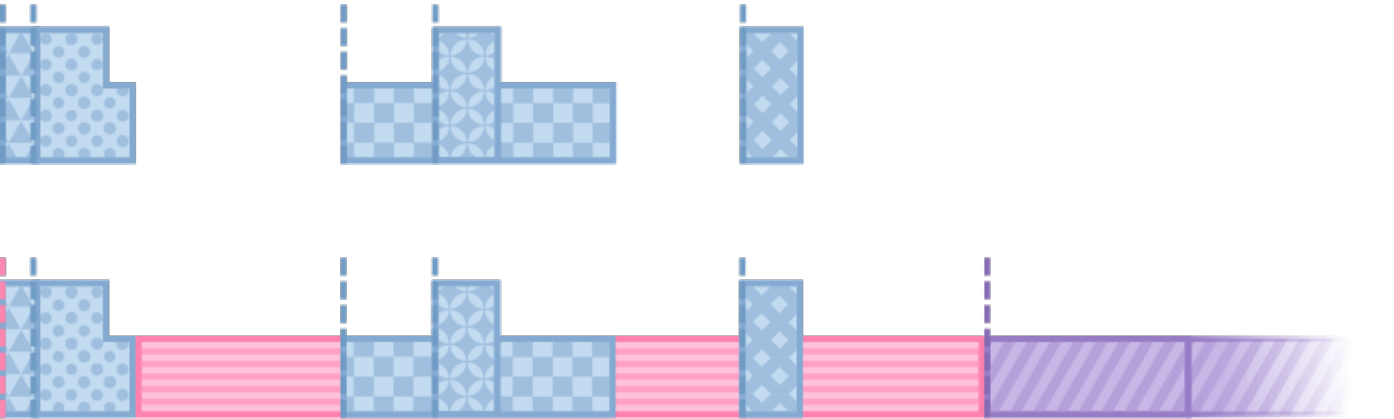}
\centering
\caption{Jobs~$n+1,\dots,2n+1$ are added to an optimal schedule for \budgetinstance\ and run at speed~$s_1$. This gives a schedule for \flowenergyinstance\ where jobs~$1,\dots,n$ have a total processing time of~$V-Y$, and are the only jobs that run at a speed higher than~$s_1$.}
\label{fig:reduction_opt_for_both}
\end{figure} 

\begin{lemma}
    The schedule~$\tilde{\sigma}^*_{\{1,\dots,n\}}$ is optimal for~$\mathcal{I}^B$.
    \label{lemma:fe-reduction}
\end{lemma}
\begin{proof}
    We first argue that~$\tilde\sigma^*_{\{1,\dots, n\}}$ is a feasible schedule. Indeed, by definition jobs~$\{1,\dots, n\}$ are scheduled identically in~$\tilde\sigma^*$ and~$\tilde\sigma^*_{\{1,\dots, n\}}$, and each such job has the same volume and release time in each of the corresponding schedules. It directly follows that each such job is completely processed after its respective release time. By Lemma~\ref{lem:budget}, it also follows that the bugdet~$B$ is not exceeded.

    Assume that~$\tilde\sigma^*_{\{1,\dots, n\}}$ is not optimal for \budgetinstance, in other words that there exists some other feasible schedule~$\sigma'$ for \budgetinstance\  with strictly less flow than~$\tilde\sigma^*_{\{1,\dots, n\}}$. We argue that this contradicts the optimality of~$\tilde\sigma^*$ for \flowenergyinstance. Indeed consider schedule~$\tilde\sigma'$ for \flowenergyinstance\ defined as follows:
    \begin{align*}
    \left( \tilde{J}'(t), \tilde{S}'(t) \right) = 
        \begin{cases}
            \left(\tilde{J}^*(t), \tilde{S}^*(t) \right) & \text{if } t\ge r_{n+2} \\
            \left(J'(t),S'(t)\right) & \text{if } J'(t)\neq \emptyset \\
            \left( n+1, \tilde{s}_1\right) & \text{at any other time } t\,.
        \end{cases}
    \end{align*}
    First note that~$\tilde\sigma'$ is defined for all times~$t$, and by construction produces feasible schedules. Also note that it uses time~$V-Y$ for jobs~$1,\dots, n$ and all other jobs are processed at a speed of~$\tilde{s}_1$. Thus, the overall energy consumption of~$\tilde\sigma'$ is the same as~$\tilde\sigma^*$. The same holds for the flow times of jobs~$n+1,\dots, 2n+1$. However, the total flow time of jobs~$1,\dots, n$ is strictly smaller under~$\tilde\sigma'$ than under~$\tilde\sigma^*$, leading to a contradiction.
\end{proof}

\UAhard*
\begin{proof}
    Given Lemma~\ref{lemma:fe-reduction}, it remains to argue that \flowenergyinstance\ can be constructed in polynomial time from \budgetinstance\ and~$\tilde\sigma^*_{\{1,\dots, n\}}$ in polynomial time from~$\tilde\sigma^*$. The first one clearly follows by construction. The second one is less obvious, given that a schedule is defined as two functions of~$t$. However, we note that for any optimal schedule it is without loss of generality to assume that each job starts processing at a speed of~$s_2$ (respectively~$\tilde{s}_2$), and only switches to~$s_1$ (respectively~$\tilde{s}_1$) at most once. Thus, the respective schedules and the associated transformation can be described in polynomial time.
\end{proof}

\begin{corollary}
    FE-ICUA is $\mathsf{NP}$-hard.
\end{corollary}

\subsection{Hardness of FE-IDWU}\label{sec:FEWU_hardness}
We show that FE-IDWU is $\mathsf{NP}$-hard through a reduction from \textsc{SubsetSum}. In~\cite[Section 3.2.2]{barcelo}, Barcelo et al. give a reduction from \textsc{SubsetSum} to B-IDWU where they show that in the B-IDWU instance, there is a gap between the optimal flows corresponding to a YES-instance and a NO-instance of \textsc{SubsetSum}. To adapt this idea for a reduction to FE-IDWU, we simulate a budget in a similar way to the reduction from B-IDUA to FE-IDUA from Section~\ref{sec:hardness_FEUA}. Rather than using a job with a comparatively large volume, we use multiple lightweight unit-sized jobs. This does introduce a new challenge, as the order of the jobs becomes harder to predict (i.e. a lightweight job could still be prioritized over a heavier job). We circumvent this hurdle by making the weight of the new jobs so lightweight that the increase in the optimal flow due to their addition cannot be too large. To be more specific: even with the inclusion of the lightweight jobs, there remains a gap between the optimal flows corresponding to a YES-instance and a NO-instance of \textsc{SubsetSum}.

Consider an instance \subsetsuminstance\ of \textsc{SubsetSum}, where we are given~$m$ elements~$a_1\geq\dots\geq a_m$ with~$a_i\in\mathbb{N}$ and a value~$A\in \mathbb{N}$ with~$a_1<A<\sum_{i=1}^m a_i$. The instance is a YES-instance if and only if there is a subset~$L\subseteq\{1,\dots,m\}$ such that~$\sum_{i\in L}a_i=A$. \textsc{SubsetSum} in known to be an $\mathsf{NP}$-hard problem. We assume that~$a_i\leq 2a_{i'}$ for all~$i,i'\in \{1,\dots,m\}$ as Barcelo et al. show that even with this assumption, \textsc{SubsetSum} remains $\mathsf{NP}$-hard~\cite[proof of Theorem 38]{barcelo}.

For our reduction to FE-IDWU, we define the instance \flowenergyinstance\ as follows. We first construct a \emph{job package}~$\J_i$ for each of the elements~$a_i$. These are identical to the packages from the reduction to B-IDWU from~\cite[Section 3.2.2]{barcelo} and defined as follows. Each~$\J_i$ consists of~$m+1$ jobs:~$\J_i = \{(i,0),\dots,(i,m)\}$. For~$i\in\{1,\dots,m\}$, job~$(i,0)$ has weight~$w_{i,0}=\frac{a_i}{m}$, and jobs~$(i,j)$ with~$j\geq 1$ have weight~$w_{i,j}=2ma_1^3$. Job~$(1,0)$ releases at~$r_{1,0}=0$ and jobs~$(i,0)$ with~$i\in\{2,\dots, m\}$ have release time~$r_{i,0}=r_{i-1,0}+m+1$. The remaining jobs in~$\J_i$ release at~$r_{i,1}=\dots=r_{i,m}=r_{i,0}+1-\alpha_i$ with~$\alpha_i = \frac{a_i}{2a_1^2}$. For ease of notation, we denote~$r_i = r_{i,1}=\dots=r_{i,m}$ and~$w_i = w_{i,1}= \dots =w_{i,m}$.

Since we want to utilize results from the reduction to B-IDWU by Barcelo et al., we want to construct an FE-IDWU instance that simulates a budget. To do this, we add two new job packages:~$\J_0$ and~$\J_{m+1}$, which serve similar roles to job~$n+1$, and jobs~$n+1,\dots,2n+1$ from the reduction in Section~\ref{sec:hardness_FEUA}.

Let~$\J_{0} = \{(0, 1),\dots, (0, K)\}$ and~$\J_{m+1}=\{(m+1,1),\dots,(m+1,\Tilde{K})\}$  with~$K=\left\lceil\frac{m^2}{2}+\frac{A}{2a_1^2}\right\rceil$ and~$\Tilde{K}=99m^8 a_1^3$.  All jobs in~$\J_0$ have release time~$r_{0}=r_{0,i}=0$ and weight~$w_{0}=w_{0,i}=\frac{1}{32m^5}$. All jobs in~$\J_{m+1}$ release at~$r_{m+1}=r_{m+1,i} =m(m+1)+ \left\lceil\frac{m^2}{2}+\frac{A}{2a_1^2}\right\rceil - \frac{m^2}{2}-\frac{A}{2a_1^2}$, and have weight~$w_{m+1}=w_{m+1,i}=\frac{1}{33m^5}$. Note that although the number of jobs depends on~$A$, all jobs in~$\J_0$ and~$\J_{m+1}$ are respectively identical and therefore the encoding length remains polynomial.

For the processor, we have two speeds:~$s_1=1$ and~$s_2=2$, with respective power consumptions~$P_1=1$ and~$P_2=3m^3a_1^3 + \frac{1}{33m^5} + 2$. This give us that $\Delta_1 = 3m^3a_1^3 + \frac{1}{33m^5}$.

\begin{observation}\label{obs:reduction_WU_lowest_weight}
    Let~$\sigma^*$ be an optimal schedule for \flowenergyinstance. Then for any~$j\in\J_{m+1}$, we have that~$s_j(\sigma^*)=s_1$ and for all~$j'\in\bigcup_{i=0}^m\J_i$, we have~$C_j(\sigma^*)< C_{j'}(\sigma^*)$.
\end{observation}
\begin{proof}
    Consider a job~$j\in\J_{m+1}$. For any~$j'\in\bigcup_{i=0}^m\J_i$, we have~$w_j<w_{j'}$ and~$r_j>r_{j'}$. Hence, it can never be optimal to complete~$j$ before~$j'$. It follows that~$j$ can also not affect any of the jobs in~$\bigcup_{i=0}^m\J_i$, from which we obtain that~$\kappa_j(\sigma^*)\leq\Tilde{K}w_{m+1} = 3m^3a_1^3 < \Delta_1$. By Corollary~\ref{cor:kappa_delta_balance}, we must have~$s_j(\sigma^*)=s_1$.
\end{proof}

For a schedule~$\sigma$, let~$F_i(\sigma) = \sum_{j\in \J_i}F_j(\sigma)$ and~$\chi_i = \sum_{j\in \J_i}x_j(\sigma)$. Furthermore, let~$\Tilde{\chi}(\sigma) = \sum_{i=0}^m \chi_i$ and~$\Tilde{F}(\sigma) = \sum_{i=0}^m F_i(\sigma)$. 

Let $Y := \frac{m^2}{2}+\frac{A}{2a_1^2}$. The following observation shows that for finding an optimal schedule, it is sufficient to only consider schedules with~$\Tilde\chi(\sigma) = m(m+1) + K - Y$.

\begin{restatable}{lemma}{WUshrinkage}\label{lemma:WU_shrinkage}
    In any optimal schedule~$\sigma^*$ for \flowenergyinstance, we have that $\Tilde{\chi}(\sigma^*) = m(m+1) + K - Y$.
\end{restatable}
\begin{proof}[Proof sketch]
   By Observation~\ref{obs:reduction_WU_lowest_weight}, having~$\Tilde{\chi}(\sigma^*) = m(m+1) + K - Y$ is equivalent to finishing the processing of all jobs in~$\bigcup_{i=0}^m\J_i$ exactly at time
   \begin{align*}
    m(m+1) + K - \left(\frac{m^2}{2}+\frac{A}{2a_1^2}\right) = r_{m+1}\,.
    \end{align*}
    
    First, suppose that~$\Tilde{\chi}(\sigma^*) > r_{m+1}$. The shrinkage among the jobs in~$\bigcup_{i=0}^m\J_i$ is strictly less than~$\frac{m^2}{2}+\frac{A}{2a_1^2}< \frac{m(m+1)}{2}$. Hence,  there must be a job~$j$ in~$\bigcup_{i=0}^m\J_i$ that does not fully run at speed~$s_2$ and thus has~$\Delta_j(\sigma^*)= \Delta_1$. Furthermore, since~$\Tilde{\chi}(\sigma^*) > r_{m+1}$, job~$j$ affects all jobs in~$\J_{m+1}$ and we have 
   \[
    \kappa_j(\sigma^*) = w_0 + \Tilde{K}w_{m+1} \geq \frac{1}{32m^5} + \frac{99m^8a_1^3}{33m^5} = \frac{1}{32m^5} + 3m^3a_1^3 > \Delta_1 = \Delta_j(\sigma^*)\,.
   \]
    By Corollary~\ref{cor:kappa_delta_balance}, this contradicts that~$\sigma^*$ is optimal.

    Now suppose that~$\Tilde{\chi}(\sigma^*) < r_{m+1}$. Then there must be some job~$j$ in~$\bigcup_{i=0}^m\J_i$ that runs faster than~$s_1$. In this case~$j$ does affect any of the jobs in~$\J_{m+1}$, and we have
   \[
    \kappa^+_j(\sigma^*) \leq Kw_0 + mw_{1,0} + m^2w_i \leq 3m^3a_1^3 < \Delta_1 = \Delta^+_j(\sigma^*)\,.
   \]
    By Corollary~\ref{cor:kappa_delta_balance}, this contradicts that~$\sigma^*$ is optimal.
\end{proof}

 Let~$\sigma_1$ be a schedule where all jobs in~$\bigcup_{i=0}^m\J_i$ run at speed~$s_1=1$. Then~$\Tilde{\chi}(\sigma_1) = m(m+1) + K$. Compared to~$\sigma_1$, in an optimal schedule the total processing time of~$\bigcup_{i=0}^m\J_i$ shrinks by a total of~$Y$. We can consider the problem as optimally ``distributing'' the total shrinkage of~$Y$ among all jobs in~$\bigcup_{i=0}^m\J_i$. Note that by Observation~\ref{obs:reduction_WU_lowest_weight}, if~$\Tilde{\chi}(\sigma^*) = m(m+1) + K - Y$ then finding an optimal schedule for~$\J_{m+1}$ is trivial. Thus, we turn our focus to optimizing the flow of~$\bigcup_{i=0}^m\J_i$.

For a schedule~$\sigma$ and~$i\in\{1,\dots,m\}$, let~$y_i(\sigma)=m+1 - \chi_i(\sigma)$. In other words, compared to its processing time at speed~$s_1$, the processing time of~$\J_i$ shrinks by~$y_i(\sigma)$ in~$\sigma$. By Lemma~\ref{lemma:WU_shrinkage}, we must have that~$\sum_{i=1}^m y_i(\sigma) \leq Y$.

Now consider a problem where among all schedules with a shrinkage of~$Y$, we want to find one that minimizes the flow of~$\bigcup_{i=1}^m \J_i$. 
It is not trivial whether an optimal schedule~$\sigma^*$ for \flowenergyinstance\ is also optimal for this problem. This would require proving that for \flowenergyinstance, it is never optimal to shrink a job in~$\J_0$ or to give it priority over a job in~$\bigcup_{i=1}^m \J_i$. Hence, we take a different approach and show that for~$\sigma^*$, the flow of~$\bigcup_{i=0}^m \J_i$ is not much larger than the optimal flow of~$\bigcup_{i=1}^m \J_i$ for a shrinkage of~$Y$.

\begin{lemma}\label{lemma:WU_flowbound}
    For \flowenergyinstance , let~$\sigma^*$ be an optimal schedule and let~$\sigma^{**}$ be a schedule that, among all schedules~$\sigma$ with~$\sum_{i=1}^m y_i(\sigma)= Y$, minimizes the flow of~$\bigcup_{i=1}^m\J_i$. Then 
   \[
    \sum^m_{i=1} F_i(\sigma^{**}) \leq \Tilde{F}(\sigma^*) \leq \sum^m_{i=1} F_i(\sigma^{**}) + \frac{1}{16m}\,.
   \]
\end{lemma}
\begin{proof}
    Since~$\sigma^{**}$ minimizes~$\sum^m_{i=1} F_i(\sigma^{**})$, we have~$\sum^m_{i=1} F_i(\sigma^{**})\leq \sum^m_{i=1} F_i(\sigma^*) \leq \Tilde{F}(\sigma^*)$.
    
    Note that~$\sigma^{**}$ will always prioritize jobs in~$\bigcup_{i=1}^m\J_i$ over jobs in~$\J_0$. It follows that w.l.o.g., we may assume all jobs in~$\J_0$ run at speed~$s_1$ and~$\sigma^{**}$ satisfies~$\Tilde{\chi}(\sigma^{**}) = m(m+1) + K - Y$. Thus, each job~$j\in\J_0$ has~$C_j(\sigma^{**})\leq r_{m+1}$, from which it follows that
   \[
    \Tilde{F}(\sigma^*) \leq \Tilde{F}(\sigma^{**}) \leq \sum^m_{i=1} F_i(\sigma^{**}) + Kr_{m+1}w_0\,.
   \]
    Since~$A < ma_1\leq m^2a_1^2$ we have that
   $
    K = \left\lceil\frac{m^2}{2}+\frac{A}{2a_1^2}\right\rceil \leq \left\lceil\frac{m^2}{2}+\frac{m^2}{2}\right\rceil = m^2
   $.
    W.l.o.g., we assume that~$m\geq 2$. Using that~$r_{m+1}= m(m+1)+K-Y$, we obtain
   \[
   \Tilde{F}(\sigma^*) \leq \sum^m_{i=1} F_i(\sigma^{**}) + m^2\frac{m(m+1)+1}{32m^5}\leq \sum^m_{i=1} F_i(\sigma^{**}) + \frac{1}{16m}\,.\qedhere
   \]
\end{proof}

We can now use Lemma~\ref{lemma:WU_flowbound} to relate the optimal flow \flowenergyinstance\  to the budget instance from the reduction to B-IDWU by Barcelo et al.~\cite{ComplexitySpeedScaling}. Consider an instance \budgetinstance\  of B-IDWU with the same two speeds:~$\Tilde{s}_1=1$ and~$\Tilde{s}_2=2$, job packages~$\Tilde{\J}_i=\J_i$ for~$i\in \{1,\dots,m\}$ (with identical weights ans release times), power consumptions~$\Tilde{P}_1=1$ and~$\Tilde{P}_2=4$, and a budget~$B=2Y$. Let~$\Tilde{\sigma}^B$ denote the schedule where all jobs~$(i,j)$ with~$j\geq 1$ run fully at speed~$s_2$, and for all~$i\in\{1,\dots,m\}$~$(i,0)$ runs at speed~$s_1$ and finishes after all other jobs in~$\J_i$. Let~$F^B = \sum_{i=1}^m F_i(\Tilde{\sigma}^B)$.

\begin{lemma}[Barcelo et al.~\protect{\cite[Theorem 38]{barcelo}}]  \label{lemma:barcelo_WU}
    Let~$\Tilde{\sigma}^*$ be an optimal schedule for \budgetinstance.
    \begin{itemize}
        \item If \subsetsuminstance\ is a YES-instance, then~$\sum_{i=1}^m F_i(\Tilde{\sigma}^*) \leq F^B -\frac{A}{2}-\frac{1}{8m}$.
        \item If \subsetsuminstance\ is a NO-instance, then~$\sum_{i=1}^m F_i(\Tilde{\sigma}^*) \geq F^B -\frac{A}{2}$.
    \end{itemize}
\end{lemma}

\begin{lemma}\label{lemma:WU_reduction}
    Let~$\sigma^*$ be an optimal schedule for \flowenergyinstance. Then~$\Tilde{F}(\sigma^*)< F^B -\frac{A}{2}$ if and only if~\subsetsuminstance\ is a YES-instance of \textsc{SubsetSum}.
\end{lemma}
\begin{proof}
    For \subsetsuminstance, we have~$\Tilde{\Delta}_1=2$, and thus a budget of~$2Y$ is equivalent to a total shrinkage of~$Y$ among the jobs in~$\bigcup_{i=1}^m\Tilde{\J}_i$. It follows that~$\sum_{i=1}^m y_i(\Tilde{\sigma}^*)=Y$. Since~$\Tilde{\J}_i=\J_i$ for all~$\{i=1,\dots,m\}$, by Lemma~\ref{lemma:WU_flowbound} we have that
   \[
    \sum^m_{i=1} F_i(\Tilde{\sigma}^*) \leq \Tilde{F}(\sigma^*) \leq \sum^m_{i=1} F_i(\Tilde{\sigma}^*) + \frac{1}{16m}\,.
   \]
    From Lemma~\ref{lemma:barcelo_WU}, it follows that if \subsetsuminstance\ is a YES-instance, then
   \[
    \Tilde{F}(\sigma^*) \leq \sum^m_{i=1} F_i(\Tilde{\sigma}^*) + \frac{1}{8m} \leq F^B -\frac{A}{2}-\frac{1}{16m}\,,
   \]
    and if \subsetsuminstance\ is a NO-instance, then
   \[
    \Tilde{F}(\sigma^*)\geq \sum_{i=1}^m F_i(\Tilde{\sigma}^*) \geq  F^B -\frac{A}{2}\,.\qedhere
   \]
\end{proof}

From Lemma~\ref{lemma:WU_flowbound}, it follows directly that FE-IDWU is $\mathsf{NP}$-hard, proving \textbf{Theorem~\ref{thm:WU}}.

\begin{corollary}
    FE-ICWU is $\mathsf{NP}$-hard.
\end{corollary}

\subsection{Fixed Priority Ordering}\label{sec:hardness-priority}
Suppose that we are given a fixed priority ordering of the jobs~$1\prec\dots\prec n$. This ordering specifies that, at any point in time and among all jobs that are released and not yet completed, the one processed is the one that comes first in the ordering. Thus, together with a speed profile~$S(t)$, a priority ordering fully specifies a schedule. This leads to the following question: \emph{Given a priority ordering such that there is an optimal schedule that adheres to this ordering, can we efficiently find an optimal schedule?}

Recall that *-****-P denotes the problem variants with a fixed priority ordering. A priority ordering is called \emph{optimal} if there is an optimal schedule that adheres to the ordering. Below we show that B-IDUA-P, FE-IDUA-P, B-IDWU, and FE-IDWU-P are $\mathsf{NP}$-hard. For B-IDUA-P, FE-IDUA-P, and B-IDWU-P, we actually prove a stronger statement: the problems are $\mathsf{NP}$-hard, even when the priority ordering is known to be optimal.

\begin{restatable}{theorem}{prioBUA}\label{thm:prio_BUA}
    B-IDUA-P is $\mathsf{NP}$-hard, even if the given ordering is known to be optimal.
\end{restatable}
\begin{proof}[Proof sketch]
    The result follows from extending the analysis of the reduction from \textsc{SubsetSum} to B-IDUA by Barcelo et al.~\cite{barcelo}. Further details can be found in Appendix~\ref{app:prio_BUA}.
\end{proof}

\begin{corollary}\label{cor:opt_prio_FEUA}
    FE-IDUA-P is $\mathsf{NP}$-hard, even if the given ordering is known to be optimal.
\end{corollary}
\begin{proof}
    Consider an instance \budgetinstance of B-IDUA-P with priority ordering~$1\prec\dots\prec n$, and assume this ordering is known to be optimal. Let \flowenergyinstance\  be the FE-IDUA instance obtained by applying the reduction from Section~\ref{sec:hardness_FEUA} to \budgetinstance. Let~$n+1,\dots,2n+1$ be the newly added jobs, as described in the reduction. By Lemma~\ref{obs:reduction_order} and the proof of Lemma~\ref{lemma:reduction_Y_on_old}, any optimal schedule for \flowenergyinstance adheres to the priority ordering~$1\prec\dots\prec n\prec n+1\prec\dots\prec 2n+1$.
    From Lemma~\ref{lemma:fe-reduction} and Theorem~\ref{thm:prio_BUA}, it follows that FE-IDUA-P is $\mathsf{NP}$-hard, even when the priority ordering is known to be optimal.
\end{proof}

\begin{theorem}\label{thm:BWU_prio}
    B-IDWU-P is $\mathsf{NP}$-hard, even if the given ordering is known to be optimal.
\end{theorem}
\begin{proof}
    For the reduction from \textsc{SubsetSum} to B-IDWU~\cite[Observation 36]{barcelo} (also described in the proof of Lemma~\ref{lemma:WU_reduction}), Barcelo shows that there is always an optimal schedule that satisfies any priority ordering with~$(i,1)\prec\dots\prec(i,m)\prec(i,0)$, for all~$i\in\{1,\dots,m\}$. We can extend this to a full priority ordering by setting~$(i,j)\prec(i',j')$ for~$i<i'$. The job packages are released far enough apart such that in any (reasonable) schedule, all jobs in a package finish before the first job in the next package is released. Hence, this full priority ordering is always optimal.
\end{proof}

\begin{theorem}
    FE-IDWU-P is $\mathsf{NP}$-hard.
\end{theorem}
\begin{proof}
    Consider an FE-IDWU-P instance \flowenergyinstance\ where we have the instance of FE-IDWU as described in Section~\ref{sec:FEWU_hardness} with a priority ordering that satisfies the following. For~$i\in\{1,\dots,m\}$, we have~$(i,1)\prec\dots\prec(i,m)\prec(i,0)$, and for any jobs~$j\in \bigcup_{i=1}^m$,~$j'\in\J_0$, and~$j''\in\J_{m+1}$, we have~$j\prec j'\prec j''$.
    
    Observation~\ref{obs:reduction_WU_lowest_weight} shows that adhering to~$j\prec j''$ and~$j'\prec j''$ for~$j\in \bigcup_{i=1}^m$,~$j'\in\J_0$, and~$j''\in\J_{m+1}$ is trivial. Since the jobs in~$\bigcup_{i=1}^m\J_i$ are prioritized over jobs in~$\J_0$, an optimal schedule for \flowenergyinstance\ must be as~$\sigma^{**}$ from Lemma~\ref{lemma:WU_flowbound}. Hence, Lemma~\ref{lemma:WU_flowbound} still holds, and by the same argument as for Theorem~\ref{thm:BWU_prio}, FE-IDWU with a fixed priority ordering is also $\mathsf{NP}$-hard.
\end{proof}

\section{An LP for Fixed Completion Ordering}
\label{sec:lp}

In this section, we investigate different variants of the problem with fixed completion-time ordering. We are given an ordering~$1\prec \dots \prec n$, of the jobs and require that~$\hat{C}_1(\sigma)\leq\dots\leq \hat{C}_n(\sigma)$, where the \emph{extended completion time}~$\hat{C}_j(\sigma)$ of job~$j$ under schedule~$\sigma$ is defined by 
\begin{align*}
    \hat{C}_1(\sigma) &:= C_1(\sigma)\\
    \hat{C}_j(\sigma) &:= \max\{C_j(\sigma), \hat{C}_{j-1}(\sigma)\} \,,\quad \forall j\in \{2,\dots,n\}\,,
\end{align*}
where~$C_j(\sigma)$ is given by the maximum~$t$ in~$X_j(\sigma)$ (i.e., the ``classic'' definition of a completion time). 
The extended flow time of each job~$j$ is given as~$\hat{F}_j(\sigma) = \hat{C}_j(\sigma)-r_j$. An intuition behind the extended completion times is although the actual processing of a job~$j$ might finish at~$C_j(\sigma)$,~$j$ can only be returned to the user once~$1,\dots,j-1$ are returned.\footnote{One can only serve desert after serving the main course, even if the desert was prepared earlier.} In this case it is natural that the flow time also accounts for the time interval~$[C_j(\sigma), \hat{C}_j(\sigma))$.

We define~$Q_j := \{j':j'\preceq j\}$,~$R_{t}^- := \{j: r_{j}\leq t\}$ and~$R_{t}^+ := \{j: r_{j}\geq t\}$. Intuitively, for time~$t$,~$R_t^+\cap Q_j$ gives a subset of the jobs that need to be fully processed before we can complete~$j$. We will show that the following LP computes the optimal schedule for FE-ID**-C with given completion ordering~$1\prec \dots \prec n$ \footnote{Since $\sum_{j} w_j \hat{F}_j + \sum_{i,j}\lambda_{j}^i E_{j}^i = \sum_{j} w_j \hat{C}_j + \sum_{i,j}\lambda_{j}^i E_{j}^i - \sum_{j} w_j {r}_j$, and $\sum_{j} w_j {r}_j$ is given and thus constant, optimizing over either objective is equivalent.}.
\begin{lpformulation}
\lpobj{min}{\sum_{j} w_j \hat{C}_j + \sum_{i,j}\lambda_{j}^i E_{j}^i}
\lpeq[lp:ordering]{\hat{C}_j\geq \hat{C}_{j-1}}
    {j>1}
\lpeq[lp:comp_times]{\hat{C}_j\geq r_{j'} + \sum_{j''\in Q_j\cap R_{r_{j'}}^+}  \sum_{i} \lambda_{j''}^i x_{j''}^i}{j,\forall j'\in Q_j\cap R_{r_j}^-}
\lpeq[lp:conv_sum]{\sum_{i} \lambda_{j}^i = 1}{j}
\lpeq[lp:conv_nonneg]{\lambda_{j}^i\geq 0}{i,j\,.}
\end{lpformulation}

We denote a feasible solution~$(\hat{C}_1,\dots, \hat{C}_n, \lambda_1^1,\dots, \lambda_n^k)$ as~$(\hat{C},\lambda)$.
Recall that by Observation~\ref{obs:proc_times_conv_comb}, for any schedule~$\sigma$ and job~$j$,~$x_j(\sigma)$ is a convex combination of the~$x_j^i$.

\begin{theorem}
    The  LP computes in polynomial time an optimal schedule for FE-ID**-C.
    \label{thm:lp}
\end{theorem}

We start by showing that any feasible schedule for a FE-ID**-C instance implies a feasible solution to the LP with the same objective value.
\begin{lemma}
    Let~$\sigma$ be a feasible schedule for FE-ID**-C with given completion ordering~$1\prec \dots \prec n$. Then there is a feasible solution~$(\hat{C},\lambda)$ to the LP such that for every~$j$:~$\hat{C}_j = \hat{C}_j(\sigma)$,~$x_j(\sigma)=\sum_i \lambda_j^i x_j^i$ and~$E_j(\sigma) = \sum_{i,j}\lambda_{j}^i E_{j}^i$.
    \label{lem:sched-lp}
\end{lemma}
\begin{proof}
    Note that by Observation~\ref{obs:proc_times_conv_comb}, we can find a~$\lambda$ which satisfies constraints \eqref{lp:conv_sum} and \eqref{lp:conv_nonneg}, and furthermore~$x_j(\sigma)=\sum_i \lambda_j^i x_j^i$ and~$E_j(\sigma) = \sum_{i,j}\lambda_{j}^i E_{j}^i$. By the definition of~$\hat{C}_j(\sigma)$ and the feasibility of~$\sigma$,
    constraint \eqref{lp:ordering} is also satisfied. 
    
    In order to show that constraint \eqref{lp:comp_times} is satisfied,  consider arbitrary~$j$,~$j'\in Q_j\cap R_{r_j}^-$ and~$j'' \in Q_j\cap R_{r_{j'}}^+$. Since~$j''\prec j$ and~$r_{j''}\ge r_{j'}$, by the feasibility of~$\sigma$, it must work on job~$j''$ for at least~$x_{j''} = \sum_i \lambda_{j''}^i x_{j''}^i$ during~$[r_{j'},\hat{C}_j(\sigma))$.
  Summing up over all such~$j''$, we obtain
    \[
        \hat{C}_j = \hat{C}_j(\sigma)\geq r_{j'} + \sum_{j''\in Q_j\cap R_{r_{j'}}^+}\sum_{i} \lambda_{j''}^i x_{j''}^i\,.\qedhere
    \]
\end{proof}

We next show that an optimal solution to the LP implies a feasible schedule with the same objective value.

\begin{lemma}
    Let~$(\hat{C},\lambda)$ be an optimal solution to the LP. Then, for FE-ID**-C 
    there is a schedule~$\sigma$ with~$\hat{C}_j(\sigma) = \hat{C}_j$,~$x_j(\sigma)=\sum_i \lambda_j^i x_j^i$, and~$E_j(\sigma) = \sum_{i,j}\lambda_{j}^i E_{j}^i$. 
    \label{lem:lp-sched}
\end{lemma}
\begin{proof}
    Given~$(\hat{C},\lambda)$, we construct a corresponding schedule~$\sigma$ as follows:
    we run each job~$j$ for a total of~$\sum_{i}\lambda_j^ix_j^i$ time (by employing the two relevant consecutive speed levels and at most one speed switch). At every time~$t$, we process the active job with the earliest~$\hat{C}_j$ (note that this is also the earliest job in the completion-time ordering). It is straightforward to express the resulting schedule in terms of $\sigma=(J(t),S(t))$.

    Schedule~$\sigma$ satisfies the last two equalities in the statement by construction. Again by construction, every job is fully processed after its release time. It remains to argue that~$\hat{C}_j(\sigma) = \hat{C}_j$, as well as the completion-time ordering is satisfied. Let~$\hat{C}_j(\sigma)$ be the extended completion time of job~$j$ in~$\sigma$. By definition, the completion-time ordering holds for~$\hat{C}_j(\sigma)$ and it only remains to prove that~$\hat{C}_j(\sigma)=\hat{C}_j$ holds for every job~$j$.

    We first argue that for each job~$j$ there holds~$\hat{C}_j(\sigma)\ge \hat{C}_j$. Note that for any~$t\le r_j$, any feasible schedule must fully process all jobs in~$Q_j\cap R_{t}^+$ within the interval~$[t,\hat{C}_j)$. In other words, for~$\sigma$, any job~$j$ and  any~$t\le r_j$ there must hold that~$\hat{C}_j(\sigma)-t \ge \sum_{j''\in Q_j\cap R_t^+} x_{j''}(\sigma) = \sum_{j''\in Q_j\cap R_t^+} \sum_i \lambda_{j''}^ix_{j''}^i$.
    So consider an arbitrary job~$j$, and 
    let~$t'\le r_j$ be the largest time  such that~$\sigma$ is either idling or processing a job~$j^*$ with~$j^* \succ j$ just before~$t'$. Note that~$t'$ must be the release time of some~$j'\in Q_j$. By the above observation, it must be the case that~$\hat{C}_j(\sigma) \ge r_{j'} + \sum_{j''\in Q_j\cap R_{r_j'}^+} \sum_i \lambda_{j''}^ix_{j''}^i$. Furthermore, by construction,~$\sigma$ never unnecessarily idles and does not process any job not in~$Q_{j''}\cap R_{r_{j'}}^+$ throughout~$[r_{j'},\hat{C}_j)$. So overall it must be the case that~$\hat{C}_j(\sigma) = r_{j'} + \sum_{j''\in Q_j\cap R_{r_{j'}}^+} \sum_i \lambda_{j''}^ix_{j''}^i$.

    If for job~$j$ constraint~\eqref{lp:comp_times} of the LP is tight, then~$\hat{C}_j(\sigma)=\hat{C}_j$ directly follows. Otherwise, by optimality of~$(\hat{C},\lambda)$, constraint~\eqref{lp:ordering} must be tight. In this case it is sufficient to show the existence of a job~$j'\prec j$ with~$\hat{C}_{j}=\hat{C}_{j'}$, such that for~$j'$ constraint~\eqref{lp:comp_times} is tight, since by the above argument~$\hat{C}_{j'}(\sigma)=\hat{C}_{j'}$ holds, and by the definition of~$\hat{C}_j(\sigma)$. The existence of such a~$j'$  follows by constraint~\eqref{lp:ordering} and the fact that for job~$1$ constraint~\eqref{lp:comp_times} must be tight.
\end{proof}

Finally, we prove the main theorem of this section.
\begin{proof}[Proof of Theorem~\ref{thm:lp}]
    Given an instance \flowenergyinstance\ of FE-ID**-C, construct the LP, solve it, construct corresponding schedule~$\sigma$ for \flowenergyinstance\ as in the proof of Lemma~\ref{lem:lp-sched}. 

    Schedule~$\sigma$ is optimal, since it has the same objective value as that of an optimal solution to the LP, and by Lemma~\ref{lem:sched-lp} the objective value of an optimal solution to the LP cannot be higher than the objective value of an optimal schedule for the corresponding input instance.

    Solving the LP and building the respective schedule can be done in polynomial time.
\end{proof}

Now suppose we change the objective of the LP to~$\min \sum_{j}w_j\hat{C}_j$ and add the  constraint~$\sum_{i,j}\lambda_j^iE_j^i\leq B$. By using the same proof, we can show that this new LP solves any B-ID**-C variant. Together with Theorem~\ref{thm:lp}, this shows us that any *-ID**-C variant can be solved in polynomial time, proving \textbf{Theorem~\ref{thm:lp-intro}}.

\section{Towards a Combinatorial Algorithm for FE-ID**-C}
\label{sec:natural-generalizations}
While the LP allows us to solve any *-ID**-C variant in polynomial time, there are still benefits in finding a  combinatorial algorithm. Barcelo et al.\cite{ComplexitySpeedScaling}  claim  such an algorithm for FE-IDUA-C by a straightforward generalization of their algorithm for FE-IDUU with only two speeds. To be more precise: they suggest that a straightforward generalization of their algorithm solves FE-IDUU with an arbitrary number of speeds, and a slightly more general result yields an algorithm for FE-IDUA-C. Unfortunately, we found that the two, in our opinion, most natural generalizations of their algorithm do not work for an arbitrary number of speeds. For more details see Appendix~\ref{app:generalization}.

Here, we provide a different and simple combinatorial algorithm for FE-IDUU with an arbitrary number of speeds. For unit-size and unit-weight jobs, ordering the jobs from earliest to latest release date (FIFO) gives an optimal completion ordering for any speed profile $S(t)$. This follows from a simple exchange argument.
Hence, the problem is equivalent to finding optimal speeds.
Algorithm~\ref{alg:kappa_delta1} is a greedy algorithm that increases the speed of a job that leads to the biggest improvement in the objective function. 
Furthermore, we conjecture a slight adaptation of this algorithm also solves any *-ID**-C variant.

\begin{algorithm}
\caption{The~$(\kappa-\Delta)$-Algorithm.\label{alg:kappa_delta1}}
\SetKwInOut{Input}{input}\SetKwInOut{Output}{output}
\While{There is a job~$j$ with~$\kappa_j\geq \Delta_j$}{
    Increase the speed of job~$j^*:=\arg\max_j (\kappa_j- \Delta_j)$
    (break ties arbitrarily), until either
    $\kappa_{j^*}- \Delta_{j^*}< 0$, or 
    $\kappa_{j^*}- \Delta_{j^*}\neq \max_j (\kappa_j- \Delta_j)$.
}
\end{algorithm}

\begin{restatable}{theorem}{KappaDelta}\label{thm:kappa_delta}
    Algorithm~\ref{alg:kappa_delta1} computes an optimal schedule for FE-IDUU in polynomial time.
\end{restatable}
\begin{proof}
    The proof can be found in Appendix \ref{app:kappa_delta}.
\end{proof}

Recall the definition of affection and note that Observation~\ref{obs:affection} does not hold for~$\hat{F}(\sigma) = \sum_j(\hat{C}_j(\sigma)-r_j)$ instead of~$F(\sigma)$. For a fixed completion ordering we define the \emph{extended} affection as follows.

\begin{definition}[Extended affection]
    Consider a schedule~$\sigma$ for a fixed completion ordering~$1\prec\dots\prec n$. To simplify notation, define~$\hat{C}_0 := -\infty$.
\begin{itemize}
    \item For~$j\leq j'$, we have that~$j'\in \hat{K}_j$ if either: 
    \begin{itemize}
        \item~$C_j(\sigma)>r_{j'}$ and~$\hat{C}_{j'}>\hat{C}_{j'-1}$, or
        \item~$\hat{C}_j>\hat{C}_{j-1}$ and~$\hat{C}_j=\hat{C}_{j'}$.
    \end{itemize}
    \item If~$j'\in \hat{K}_j$ and~$j''\in\hat{K}_{j'}$, then~$j''\in\hat{K}_j$.
    \item The extended affection of job~$j$ is defined as~$\hat\kappa_j(\sigma) := \sum_{j'\in \hat{K}_j(\sigma)} w_{j'}$.
\end{itemize} 
\end{definition}

\begin{conjecture}
    The algorithm obtained by replacing $\kappa$ with $\hat\kappa$ in Algorithm~\ref{alg:kappa_delta1} computes an optimal schedule for any FE-ID**-C variant in polynomial time.
\end{conjecture}

\section{Discussion and Future Work}
We showed the $\mathsf{NP}$-hardness of FE-IDUA, FE-ICUA, FE-IDWU and FE-ICWU, and furthermore that all these problem variants are solvable in polynomial time when given a completion-time ordering. 

Recently, there has been renewed interest in approximation algorithms for the fixed-speed case of minimizing weighted flow for arbitrary-size jobs. This line of work has culminated in a Polynomial-Time Approximation Scheme (PTAS)~\cite{ArmbrusterRW23} (recall that~$\star$-I$\star$WA is $\mathsf{NP}$-hard~\cite{LABETOULLE1984245}). In light of this, it would be interesting to study the exact approximability of the $\mathsf{NP}$-hard variants on speed-scalable processors. Our linear program for the problem given a completion time order might be a useful starting point, as it implies that it suffices to -- in some sense -- approximate the optimal completion-time ordering. 

In particular, for FE-IDUA, there exists an online~$(2+\epsilon)$-competitive algorithm (see~\cite{AndrewLW10,BansalCP13}). While this naturally implies an offline~$(2+\epsilon)$-approximation algorithm, it seems plausible that full knowledge of all jobs and their characteristics could lead to better approximation guarantees.

\bibliography{sources}

\newpage
\appendix

\section{Convex Combinations of Processing Times}\label{app:problem_desc}

\obsConv*
\begin{proof}
    First, note that if $x_j(\sigma)=x_j^i$ for some $i\in\{1,\dots,k\}$,  the result follows directly from Observation \ref{obs:no_conv_comb}. Hence, we may assume $x_j^{i-1}<x_j(\sigma)<x_j^i$ for some $i\in\{2,\dots,k\}$, which also implies that for any convex combination $x_j(\sigma)=\sum_{i=1}^k\lambda_ix_j^i$, there are at least two $\lambda_i$ that are non-zero.
    
    Consider a schedule~$\sigma$ and job~$j$ with~$s_i\leq s_j(\sigma)\leq s_{i+1}$. Let~$\lambda\in[0,1]$ be such that~$x_j(\sigma) =  \lambda x_j^{i} + (1-\lambda) x_j^{i+1}$. By Observation \ref{obs:proc_times_conv_comb}, we can achieve this processing time by running at speed~$s_i$ for~$\lambda x_j^{i}$ time, and at~$s_{i+1}$ for~$(1-\lambda) x_j^{i+1}$, which gives an energy consumption of~$\lambda E_j^{i} + (1-\lambda) E_j^{i+1}$. We now prove that any other convex combination of~$\{x_1,\dots,x_k\}$ that yields~$x_j(\sigma)$ will have a higher energy consumption than~$\lambda E_j^{i} + (1-\lambda) E_j^{i+1}$. 
    
    Let~$x_j(\sigma)=\sum_{i'=1}^k\lambda_{i'} x_j^{i'}$ be a convex combination such that the energy consumption~$\sum_{i'=1}^k\lambda_{i'} E_j^{i'}$ is the lowest among all possible convex combinations that yield~$x_j(\sigma)$.  Among all indices~$i'$ with~$\lambda_{i'}>0$, let~$q$ and~$r$ be the minimum and maximum indices, respectively. 
    
    Suppose, for the sake of contradiction, that either~$q< i$ or~$r> i+1$. We here only argue that~$q< i$, since the other case can be handled with the same argument. Then~$x_j^{q}> x_j^i >x_j^{r}$, and thus there is a~$\lambda'\in(0,1)$ such that~$x_j^i = \lambda' x_j^{q}+ (1-\lambda') x_j^{r}$. Let~$\Tilde{\lambda} = \min\{\lambda_{q}, \lambda_{r}\}$. Then we can rewrite
    \begin{align*}
        \lambda_{q}x_j^{q} + \lambda_{r}x_j^{r} 
        = \left(\lambda_{q}-\Tilde{\lambda}\lambda'\right)x_j^{q} + \Tilde{\lambda}\lambda'x_j^{q} + \left(\lambda_{r}-\Tilde{\lambda}(1-\lambda')\right)x_j^{r} + \Tilde{\lambda}(1-\lambda')x_j^{r}
    \end{align*}
    where by choice of~$\Tilde{\lambda}$, we have~$\lambda_{r}-\Tilde{\lambda}\lambda'>0$ and~$\lambda_{r}-\Tilde{\lambda}(1-\lambda')>0$. Since~$x_j(\sigma) =  \lambda' x_j^{i} + (1-\lambda') x_j^{i+1}$, it follows that 
    \begin{align*}
        x_j(\sigma)=\sum_{i'=1}^k\lambda_{i'} x_j^{i'} - \Tilde{\lambda}\left(\lambda'x_j^{q} + (1-\lambda')x_j^{r}\right) + \Tilde{\lambda}x_j^i
    \end{align*}
    is a different convex combination of~$x_j(\sigma)$, which yields an energy consumption of
    \begin{align*}
        \sum_{i'=1}^k\lambda_{i'} E_j^{i'} - \Tilde{\lambda}\left(\lambda'E_j^{q} + (1-\lambda')E_j^{r}\right) + \Tilde{\lambda}E_j^i\,.
    \end{align*}
    However, by Observation \ref{obs:no_conv_comb}, we have that $E_j^i< \lambda'E_j^{q} + (1-\lambda')E_j^{r}$, which contradicts that~$\lambda E_j^{i} + (1-\lambda) E_j^{i+1}$ is the lowest possible energy consumption for any convex combination that yields~$x_j(\sigma)$.
\end{proof}

\section{Hardness of B-IDUA-P}\label{app:prio_BUA}
This section proves the following theorem:
\prioBUA*

Consider the reduction from \textsc{SubsetSum} to B-IDUA by Bacelo et al.~\cite[Section 3]{ComplexitySpeedScaling}, which is defined as follows:

Let \subsetsuminstance\  be an instance of \textsc{SubsetSum} where given elements~$a_1\geq \dots \geq a_m$ with~$a_i\in\mathbb{N}$ and a target value~$A$ with~$a_1<A<\sum_{i=1}^m a_i$. \subsetsuminstance\ is a YES-instance if and only if there is a subset~$L\subseteq\{1,\dots,m\}$ such that~$\sum_{i\in L}a_i=A$ (otherwise, it is a NO-instance). 

Construct the following instance \budgetinstance\ of B-IDUA. For each element~$a_i$, add a job package~$\J_i=\{(i,1),(i,2)\}$. It is assumed that the release times of the job packages are shifted enough so jobs of different packages do not interact (i.e., even at speed~$s_1$, the jobs in~$\J_i$ complete before the first release time in~$\J_{i+1}$). For ease of notation, further assume that the first job in~$\J_i$ is released at time 0. Job~$(i,1)$ has release time~$r_{i,0}=0$ and processing volume~$v_{i,1}=a_i$, and job~$(i,2)$ has release time~$r_{i,2}=\frac{a_i}{2}$ and processing volume~$2\delta a_i$ with~$\delta=\frac{1}{a_1m^2}$. The allowable speeds are~$s_1=1$ and~$s_2=2$, and power consumptions~$P_1=1$ and~$P_2=4$. This gives us~$\Delta_1=2$. Finally, the budget is $B=(1 + 4\delta)\sum_{i=1}^m a_i + A$.

For a schedule~$\sigma$, let~$C_{i,j}$ denote the completion time of job~$(i,j)$. Furthermore, let~$F_i(\sigma)$ and~$E_i(\sigma)$ denote the flow and energy consumption of~$\J_i$.

The following observations appear implicitly in \cite[proof of Theorem 34]{ComplexitySpeedScaling}.

\begin{observation}\label{obs:BUA_order}
    Let~$\sigma^*$ be an optimal schedule for \budgetinstance.
    \begin{itemize}
        \item If~$E_i(\sigma^*)\in[a_i + 2\delta a_i, 2a_i]$, then~$C_{i,2}<C_{i,1}$
        \item If~$E_i(\sigma^*)\in (2a_i,2a_i+ 2\delta a_i)$, then~$C_{i,1}<C_{i,2}$~$C_{i,1}>r_{i,2}$.
        \item If~$E_i(\sigma^*)\in [2a_i+ 2\delta a_i,2a_i+ 4\delta a_i]$, then~$C_{i,1}\leq r_{i,2}$.
    \end{itemize}
\end{observation}

Note that for~$E_i(\sigma^*)= 2a_i$, it does not matter whether~$(i,1)$ or~$(i,2)$ completes first, as both yield the same flow. Hence, it is assumed that~$(i,2)$ completes first.

\begin{definition}[(Base schedule)]
    The base schedule of~$\J_i$, denoted~$\BS_i$, completes job~$(i,1)$ after~$(i,2)$, runs~$(i,1)$ at~$s_1$ and~$(i,2)$ at~$s_2$. The base schedule has an energy consumption of~$E_i(\BS_i) = a_i +4\delta a_i$ and a flow of~$F_i(\BS_i) = a_i +2\delta a_i$.
\end{definition}

\begin{observation}\label{obs:BUA_increase}
Consider two schedules~$\sigma$ and~$\sigma'$ with~$E_i(\sigma')-E_i(\sigma)=\varepsilon>0$. Assuming we always order the jobs in an optimal way, we have the following:
    \begin{itemize}
        \item If~$E_i(\BS_i)\leq E_i(\sigma)< E_i(\sigma') \leq 2a_i$, then~$F_i(\sigma') = F_i(\sigma)- \frac{\varepsilon}{2}$.
        \item If~$2a_i\leq E_i(\sigma)< E_i(\sigma') \leq 2a_i+ 2\delta a_i$, then~$F_i(\sigma') = F_i(\sigma)- \varepsilon$.
        \item If~$2a_i+ 2\delta a_i\leq E_i(\sigma)< E_i(\sigma') \leq 2a_i+ 4\delta a_i$, then~$F_i(\sigma') = F_i(\sigma)- \frac{\varepsilon}{2}$
    \end{itemize}
\end{observation}

To prove Theorem~\ref{thm:prio_BUA}, we first show that for \budgetinstance, one can predict an optimal priority ordering. Note that since the job packages are independent, it is irrelevant whether or not a job has priority over a job from a different package.

\begin{lemma}\label{lemma:prio_BUA_opt}
    Let~$\sigma^*$ be an optimal schedule for \budgetinstance. Then~$\sigma^*$ adheres to any priority ordering with~$(i,2)\prec(i,1)$ for all~$i\in\{1,\dots,m\}$.
\end{lemma}
\begin{proof}    
    First, note that by Observation~\ref{obs:BUA_order}, if~$E_i(\sigma^*)\notin (2a_i,2a_i+ 2\delta a_i)$, then~$\sigma^*$ adheres to the priority ordering~$(i,2)\prec(i,1)$. 
    We show that having~$E_i(\sigma^*)\in [2a_i,2a_i+ 2\delta a_i)$ contradicts the optimality of~$\sigma^*$. 
       
    Let~$L\subseteq \{1,\dots,m\}$ be such that in~$E_i(\sigma^*)\geq 2a_i$ if and only if~$i\in L$. In~\cite[proof of Lemma 1]{ComplexitySpeedScaling} it is shown that we must have~$\sum_{i\in L'}a_i\leq A$, otherwise the total energy consumption exceeds the budget~$B$.

    Suppose there is some~$i'\in L$ with~$E_{i'}(\sigma^*)\in [2a_{i'},2a_{i'}+ 2\delta a_{i'})$ and let~$\varepsilon= 2a_{i'}+ 2\delta a_{i'}-E_{i'}(\sigma^*)$. If for all~$i\in L$, we have~$E_i(\sigma^*)\leq 2a_i+ 2\delta a_i$, and for all~$i\notin L$, we have~$E_i(\sigma^*)\leq E_i(\BS_i)$, then
    \begin{align*}
        E(\sigma^*)&\leq \sum_{i\in L} (2a_i + 2\delta a_i) + \sum_{i\notin L} (a_i +4\delta a_i) -\varepsilon \\
        &\leq (1+4\delta)\sum_{i=1}^m a_i + \sum_{i\notin L}a_i - 2\delta a_i -\varepsilon\leq B - \sum_{i\in L} 2\delta a_i - \varepsilon\,.
    \end{align*}
    Hence, there must be some job package~$\J_{i''}$ with either~$E_{i''}(\sigma^*)\in(E_i(\BS_i),2a_i)$ or~$E_{i''}(\sigma^*)\in(2a_i+ 2\delta a_i,2a_i+ 4\delta a_i]$. From Observation~\ref{obs:BUA_increase}, it follows that there must be some~$0<\varepsilon'\leq \varepsilon$, such that decreasing the energy consumption of~$\J_{i''}$ and increasing the energy consumption of~$\J_{i'}$ by~$\varepsilon'$ gives an improvement in the flow of~$\frac{\varepsilon'}{2}$. This contradicts the optimality of~$\sigma^*$, and thus we must have that~$E_i(\sigma^*)\notin (2a_i,2a_i+ 2\delta a_i)$ for all~$i\in \{1,\dots,m\}$. It follows that~$\sigma^*$ satisfies any priority ordering with~$(i,2)\prec(i,1)$ for each~$\J_i$.
\end{proof}

\prioBUA*
\begin{proof}[Proof of Theorem~\ref{thm:prio_BUA}]
    Let~$\sigma^*$ be an optimal schedule for \budgetinstance. In~\cite{ComplexitySpeedScaling} it is shown that~$F(\sigma^*) \leq \sum_{i=1}^m F_i(\BS_i) -(\frac{1}{2}+\delta)A$ if and only if \subsetsuminstance is a YES-instance of \textsc{SubsetSum}. By Lemma~\ref{lemma:prio_BUA_opt}, it follows that this also holds for \budgetinstance\ with a fixed priority order where~$(i,2)\prec(i,1)$ for all~$i\in\{1,\dots,m\}$. We conclude that B-IDUA-P is $\mathsf{NP}$-hard, even when the priority ordering is known to be optimal.
\end{proof}

\section{Towards a Combinatorial Algorithm for FE-ID**-C}

Consider an instance of FE-IDUU with jobs~$1,\dots,n$ and speeds~$s_1,\dots,s_k$. We assume that~$r_1\leq \dots \leq r_n$. Since we have unit-size and unit-weight jobs, ordering the jobs from earliest to latest release date (FIFO) gives an optimal completion ordering for any speed profile $S(t)$. This follows from a simple exchange argument.
Hence, the problem is equivalent to finding optimal speeds for completion ordering $1\prec \dots \prec n$. In this section, we use only job speeds to describe a schedule.

\subsection{Unsuccessful FE-IDUU-Algorithm Generalizations to 3+ Speeds}\label{app:generalization}
Barcelo at al.~\cite{ComplexitySpeedScaling} show that the following algorithm computes optimal schedules for this variant. For ease of notation, we use~$\kappa_j$ and~$\Delta_j$ to denote the affection and shrinking energy of~$j$ during the algorithm (even though these values change throughout the algorithm).

\begin{algorithm}
\caption{An algorithm for FE-IDUU with two speeds (Barcelo et al.~\cite{ComplexitySpeedScaling}).\label{alg:UU_2speeds}}
\For{$j\in\{1,\dots,n\}$}{
    \While{$\kappa_j\geq \Delta_j$}{
        Increase the speed of~$j$.\
    }
}
\end{algorithm}

Note that, once a job~$j$ reaches speed~$s_2$, we have~$\Delta_j=\infty$, and thus the algorithm will not speed~$j$ up any further. 

Now consider an instance of FE-IDUU with an arbitrary number $k$ of speeds, say~$s_1,\dots,s_k$. In our opinion the following are the most natural generalizations of the two-speed algorithm.  In the first one, simply run Algorithm~\ref{alg:UU_2speeds}. This means that, even if a job~$j$ fully runs at~$s_2$, and~$\Delta_j$ has increased to~$\Delta_2$, the algorithm increases the speed of~$j$ as long as~$\kappa_j\geq \Delta_j$ holds.

For the second generalization, consider Algorithm~\ref{alg:generalization2}.

\begin{algorithm}
\caption{A naive generalization of Algorithm~\ref{alg:UU_2speeds}.\label{alg:UU_2speeds_gen}}
\label{alg:generalization2}
\For{$i\in\{1,\dots,k\}$}{
    \For{$j\in\{1,\dots,n\}$}{
    \While{$\kappa_j\geq \Delta_i$ and~$\Delta_j\leq\Delta_i$}{
        Increase the speed of~$j$.\
        }
    }
}
\end{algorithm}

In this case, we never increase the speed of a job beyond the speed~$s_i$ for the current value~$i$. So while Algorithm~\ref{alg:UU_2speeds} considers all the jobs exactly once, Algorithm~\ref{alg:UU_2speeds_gen} considers each job once for every speed.

To show that neither algorithm computes optimal schedules for FE-IDUU with more than two speeds,  consider the following instance. Let jobs~$\{1,2,3\}$  have release times~$r_1 = 0$,~$r_2=\frac{1}{3}$ and~$r_3=\frac{4}{3}$. Let~$s_1=1$,~$s_2=3$, and~$s_3=3$ be the available speeds with power consumptions~$P_1=1$,~$P_2=3+\alpha$, and~$P_3=6+\alpha$ for some~$0\leq \alpha\leq 1$. This gives us~$\Delta_1 = 1+\alpha$ and~$\Delta_2=3-\alpha$.

For Algorithm~\ref{alg:UU_2speeds}, note that while the speed of job 1 is lower than~$s_3$, it completes after~$r_2$. Furthermore, if job 2 runs at~$s_1$, then job 2 completes after~$r_3$. Hence, until~$j$ is fully sped up to~$s_3$, we have~$\kappa_1 =3 \geq \Delta_2$. Once it runs at speed~$s_3$, it completes exactly at~$r_2$, and job 2 completes exactly at~$r_3$. Thus, we have~$\kappa_1 = \kappa_2=\kappa_3=1$ and the algorithm will terminate. The resulting schedule has total  flow of~$2+\frac{1}{3}$ and energy consumption of~$4+\frac{\alpha}{3}$. This gives an objective value of
\begin{align*}
\Obj_1 = 6+\frac{2+2\alpha}{6}\,.
\end{align*}

For Algorithm~\ref{alg:UU_2speeds_gen}, job 1 is first sped up to~$s_2$, after which we move on to job 2. Before being sped up, job 2 completes at~$\frac{3}{2}= r_3 +\frac{1}{6}$. Thus, we have that~$\kappa_j=2\geq \Delta_1$ and will continue to speed up job 2 until it completes at~$r_3$. At this point, we have that~$\kappa_1=2$ and~$\kappa_2=\kappa_1 =1$, and the algorithm will terminate. The resulting schedule has total flow of~$2+\frac{1}{2}$ and energy consumption of~$3+\frac{2+2\alpha}{3}$.This gives an objective value of
\begin{align*}
\Obj_2 = 6+\frac{1+4\alpha}{6}\,.
\end{align*}
If~$\alpha<\frac{1}{2}$, then~$\Obj_1>\Obj_2$, and if~$\alpha>\frac{1}{2}$, then~$\Obj_1<\Obj_2$. It follows that neither algorithm computes optimal schedules for FE-IDUU with an arbitrary number of speeds.

\subsection{A Combinatorial Algorithm for FE-IDUU}\label{app:kappa_delta}
In this section, we prove Theorem \ref{alg:kappa_delta}. Recall that, since a FIFO ordering is always optimal, we use only job speeds to describe a schedule. We claim that initializing every job at speed~\(s_1\), and then running the following algorithm computes an optimal schedule.

\setcounter{algocf}{0}
\begin{algorithm}
\caption{The~$(\kappa-\Delta)$-Algorithm.\label{alg:kappa_delta}}
\SetKwInOut{Input}{input}\SetKwInOut{Output}{output}
\While{There is a job~$j$ with~$\kappa_j\geq \Delta_j$}{
    Increase the speed of job~$j^*:=\arg\max_j (\kappa_j- \Delta_j)$
    (break ties arbitrarily), until either
    $\kappa_{j^*}- \Delta_{j^*}\leq 0$, or 
    $\kappa_{j^*}- \Delta_{j^*}\neq \max_j (\kappa_j- \Delta_j)$.
}
\end{algorithm}

\KappaDelta*

To prove the correctness of Algorithm \ref{alg:kappa_delta}, we show that for any schedule that was not 
constructed by the algorithm, we can do the following. By slowing down one job and speeding up another, 
we either obtain an improvement in the objective, or a schedule that is ``more similar'' to the schedule 
constructed by Algorithm \ref{alg:kappa_delta}. To show that such an exchange exists, we exploit the 
fact that affection is easier to describe in a setting with unit-size and unweighted jobs. To use this, 
we introduce new terminology and prove several results on the affection in this setting. 
Algorithm \ref{alg:kappa_delta} starts with all jobs at speed~\(s_1\) and then increases the speed of the 
jobs in a specific order. To better compare schedules to the one from the algorithm, we define a way to 
construct schedules that is similar to the algorithm, but can be used to reconstruct any given schedule. 
We use this construction to determine which jobs we need for the exchange.

Since the jobs are unit-size and are ordered from earliest to latest release time, jobs are never preempted. 
Hence, the affection set of a job~\(j\) is always of the form~\(\{j,j+1,\dots,j+\ell\}\) where we 
have~\(\ell\in\{0,\dots,n-1\}\). The following definition captures this and will help us describe affection 
throughout the algorithm.

\begin{definition}[Affection chain]
    For a schedule~\(\sigma\), an \emph{affection chain} is a maximal sequence of jobs~\((j,\dots, j+\ell)\) such that job $j$ affects job $j+\ell$. 
\end{definition}
A \emph{lower affection chain} is defined analogously with lower affection. Note that for a schedule~\(\sigma\) and job~\(j\), there is a unique (lower) affection chain that contains~\(j\). We denote this chain by~\(\K_j(\sigma)\) (and~\(\K^+_j(\sigma)\) for the lower affection). Hence, for jobs~\(j,j'\in\{1,\dots,n\}\), we have that~\(j'\in\K_j(\sigma)\) implies~\(j'\in \K_j^+(\sigma)\).

\begin{observation}\label{obs:affection_from_chain}
    For a schedule~\(\sigma\) and job~\(j\), let~\(j'\) and~\(j^+\) be the last jobs in~\(\K_j(\sigma)\) and~\(\K^+_j(\sigma)\) respectively. Then~\(\kappa_{j}(\sigma)= j'-j+1\) and~\(\kappa^+_{j}(\sigma)= j^+-j+1\).
\end{observation}

\begin{observation}\label{obs:overlap_chain}
    For a schedule~\(\sigma\) and jobs~\(j_1,j_2\), if~\(\K_{j_1}(\sigma)= \K_{j_2}(\sigma)\), then 
    \[
    \kappa_{j_1}(\sigma)-\kappa_{j_2}(\sigma)  = j_2-j_1\,.
    \]
    If~\(\K^+_{j_1}(\sigma)= \K^+_{j_2}(\sigma)\), then 
    \[
    \kappa^+_{j_1}(\sigma)-\kappa^+_{j_2}(\sigma) = j_2-j_1\,.
    \]
\end{observation}
\begin{proof}
    The statement follows directly from Observation \ref{obs:affection_from_chain}.
\end{proof}

Algorithm \ref{alg:kappa_delta} starts with all jobs at the lowest speed and then only speeds them up (it never reduces the speed of any job). Hence, we start with a number of affection chains, which can only ``break'' into smaller chains during the algorithm. The following observation shows that the affection of jobs earlier in the chain will change the most when such a break occurs.

\begin{observation}\label{obs:earlier_loses_more}
    For schedules~\(\sigma\) and~\(\sigma'\) and jobs~\(j_1<j_2\), if~\(\K_{j_1}(\sigma)=\K_{j_2}(\sigma)\), then 
    \[
    \kappa_{j_1}(\sigma)- \kappa_{j_1}(\sigma')  \geq \kappa_{j_2}(\sigma)- \kappa_{j_2}(\sigma')\,.
    \]
\end{observation}
\begin{proof}
    Let~\(j^*\) be the last job in~\(\K_{j_1}(\sigma)=\K_{j_2}(\sigma)\). Then by Observation \ref{obs:affection_from_chain}, we have that~\(\kappa_{j_1}(\sigma)=j^*-j_1+1\) and~\(\kappa_{j_2}(\sigma)=j^*-j_2+1\). Let~\(j'\) be the last job in~\(\K_{j_1}(\sigma')\). If~\(j'\geq j_2\), then~\(\K_{j_1}(\sigma')= \K_{j_2}(\sigma')\) and the result follows from Observation \ref{obs:overlap_chain}. If~\(j'<j_2\), then~\(\kappa_{j_1}(\sigma')< j_2 - j_1+1\). Since~\(\kappa_{j_2}(\sigma')\geq 1\), we obtain that
    \[
    \kappa_{j_1}(\sigma)- \kappa_{j_1}(\sigma') > j^*-j_2 \geq \kappa_{j_2}(\sigma)- \kappa_{j_2}(\sigma')\,.\qedhere
    \]
\end{proof}

Suppose we extend the processing time of one job and shrink another by the same amount. In an optimal schedule, such an exchange cannot improve the objective value. Hence, if we can find an exchange that leads to an improvement, then we can prove a schedule is not optimal. To better describe this, we introduce the following notation.

Consider some schedule~\(\sigma\) and jobs~\(j_1,j_2\). For any~\(\varepsilon>0\), let~\(\sigma^{+\varepsilon}_{j_1}\) be a schedule with~\(x_j(\sigma^{+\varepsilon}_{j_1})= x_j(\sigma)\) for all~\(j\neq j_1\) and~\(x_{j_1}(\sigma^{+\varepsilon}_{j_1})= x_{j_1}(\sigma)+\varepsilon\). Let~\(\sigma^{-\varepsilon}_{j_2}\) be a schedule with~\(x_j(\sigma^{-\varepsilon}_{j_2})= x_j(\sigma)\) for all~\(j\neq j_2\) and~\(x_{j_2}(\sigma^{-\varepsilon}_{j_2})= x_{j_2}(\sigma)-\varepsilon\). Finally, let~\(\sigma_{j_1 j_2}^\varepsilon\) be the schedule with 
\begin{itemize}
    \item~\(x_j(\sigma_{j_1 j_2}^\varepsilon)= x_j(\sigma)\) for all~\(j\) with~\(j\neq j_1\) and~\(j\neq j_2\),
    \item~\(x_{j_1}(\sigma_{j_1 j_2}^\varepsilon)= x_{j_1}(\sigma)+\varepsilon\)
    \item~\(x_{j_2}(\sigma_{j_1 j_2}^\varepsilon)= x_{j_2}(\sigma)-\varepsilon\).
\end{itemize}

\begin{observation}\label{obs:upper_becomes_lower}
    For any schedule~\(\sigma\) and job~\(j\), there is an~\(\varepsilon>0\) such that~\(\kappa_j(\sigma_j^{+\varepsilon})= \kappa_j^+(\sigma)\).
\end{observation}
\begin{proof}
    Let~\(j+\ell\) be the last job in~\(\K^+_j(\sigma)\), and thus $\kappa_j^+(\sigma)=\ell+1$. 
    If~\(j+\ell=n\), then the statement holds for any~\(\varepsilon>0\). 
    Otherwise, by the definition of lower affection, we have~\(C_{j+i}(\sigma)\geq r_{j+i+1}\) for all~\(i\in\{0,\dots,\ell-1\}\) and~\(C_{j+\ell}(\sigma)<r_{j+\ell+1}\). 
    Let
    \[
    \varepsilon\leq r_{j+\ell+1}-C_{j+\ell}(\sigma)\,,
    \] 
    then~\(C_{j+i}(\sigma_j^{+\varepsilon})=r_{j+\ell+1} \geq r_{j+i+1}+\varepsilon>r_{j+i+1}\) for all~\(i\in\{0,\dots,\ell-1\}\), and~\(C_{j+\ell}(\sigma_j^{+\varepsilon})\leq r_{j+\ell+1}\). Thus,~\(\kappa_j(\sigma_j^{+\varepsilon})=\ell+1= \kappa_j^+(\sigma)\). 
\end{proof}

The affection of job~\(j_2\) in~\(\sigma^{+\varepsilon}_{j_1}\) may differ compared to~\(\sigma\). 
Hence, describing how the objective changes from~\(\sigma\) to~\(\sigma_{j_1 j_2}^\varepsilon\) is 
not always straightforward. The following lemma helps us identify when~\(\sigma_{j_1 j_2}^\varepsilon\) 
gives an improvement. For a schedule~\(\sigma\), let~\(\Obj(\sigma)\) denote the objective value of~\(\sigma\). 

\begin{lemma}\label{lemma:exchange_cases}
    Consider some schedule~\(\sigma\) and jobs~\(j_1,j_2\). If at least one of the following conditions holds:
    \begin{itemize}
        \item \(\K_{j_1}(\sigma)=\K_{j_2}(\sigma)\), or
        \item \(j_1<j_2\) and~\(\K^+_{j_1}(\sigma)=\K^+_{j_2}(\sigma)\),
    \end{itemize}
    then there is an~\(\varepsilon>0\) such that
    \[
    \Obj(\sigma^\varepsilon_{j_1j_2}) = \Obj(\sigma) - \varepsilon\left(\left(\kappa^+_{j_2}(\sigma)-\Delta_{j_2}(\sigma)\right) - \left(\kappa^+_{j_1}(\sigma)-\Delta^+_{j_1}(\sigma)\right)\right)\,.
    \]
\end{lemma}
\begin{proof}
    First, note that if~\(\K_{j_1}(\sigma)=\K_{j_2}(\sigma)\), then~\(j_1\) and~\(j_2\) remain in the same affection chain in~\(\sigma_{j_1}^{+\varepsilon}\) for any~\(\varepsilon>0\). Next, assume that~\(j_1<j_2\) and~\(\K^+_{j_1}(\sigma)=\K^+_{j_2}(\sigma)\). By Observation \ref{obs:upper_becomes_lower} we can chose an~\(\varepsilon'>0\) small enough such that~\(
    \kappa_{j_1}(\sigma_{j_1}^{+\varepsilon'}) = \kappa_{j_1}^+(\sigma)\). Since~\(j_1<j_2\), it follows that~\(j_2\in\K_{j_1}(\sigma_{j_1}^{+\varepsilon'})\).
     
     Thus, if one of the two conditions holds, we have that~\(\K_{j_1}(\sigma_{j_1}^{+\varepsilon'}) = \K_{j_2}(\sigma_{j_1}^{+\varepsilon'})\). By Observation \ref{obs:overlap_chain}, and using that~\(\kappa_{j_1}(\sigma_{j_1}^{+\varepsilon'}) = \kappa_{j_1}^+(\sigma)\), we obtain that 
    \[
    \kappa^+_{j_1}(\sigma) - \kappa^+_{j_2}(\sigma) = \kappa^+_{j_1}(\sigma) - \kappa_{j_2}(\sigma_{j_1}^{+\varepsilon'})\,.
    \]
    From Observation \ref{obs:affection} and Lemma \ref{lemma:delta}, we obtain that there must be an~\(0<\varepsilon\leq \varepsilon'\) such that
    \begin{align*}
        \Obj(\sigma^{\varepsilon}_{j_1j_2}) &= \Obj(\sigma) + \varepsilon\left(\kappa_{j_1}^+(\sigma)- \Delta_{j_1}^+(\sigma)\right) - \varepsilon\left(\kappa_{j_2}(\sigma_{j_1}^{+\varepsilon'})- \Delta_{j_2}(\sigma_{j_1}^{+\varepsilon'})\right) \\
        &= \Obj(\sigma) - \varepsilon\left(\left(\kappa^+_{j_2}(\sigma)-\Delta_{j_2}(\sigma)\right) - \left(\kappa^+_{j_1}(\sigma)-\Delta^+_{j_1}(\sigma)\right)\right)\,. \qedhere
    \end{align*}
\end{proof}

\subsubsection*{Schedule construction}
Note that any schedule~\(\sigma\) can be constructed by initializing every job at speed~\(s_1\), and then speeding up the jobs until every job~\(j\) runs at speed~\(s_j(\sigma)\). We could do these speed-ups in many different orders, but to better compare a schedule to the schedule obtained by Algorithm \ref{alg:kappa_delta}, we define a way to construct a schedule that speeds up the jobs in an order similar to the algorithm. 

A \emph{construction} of a schedule~\(\sigma\) consists of several steps~\(1,\dots,\tau\) with corresponding intermediate schedules~\((\sigma^0,\dots,\sigma^T)\) that satisfy the following
\begin{enumerate}
    \item For all~\(j\in\{1,\dots,n\}\), we have that~\(s_j(\sigma^0)=s_1\).
    \item~\(\sigma^T=\sigma\)
    \item For all~\(\tau\in\{1,\dots,T\}\), there is exactly one~\(j^\tau\in\{1,\dots,n\}\) with~\(s_{j^\tau}(\sigma^{\tau-1})<s_{j^\tau}(\sigma^\tau)\). For all~\(j\neq j^\tau\), we have that~\(s_{j}(\sigma^{\tau-1})=s_{j}(\sigma^\tau)\). We say that we speed up~\(j^\tau\) in step~\(\tau\) of the construction. \label{def:constr_one_job}
    \item For all~\(\tau\in\{1,\dots,T\}\), we have that~\(\kappa_{j^\tau}(\sigma^{\tau-1})=\kappa^+_{j^\tau}(\sigma^{\tau})\) and~\(\Delta_{j^\tau}(\sigma^{\tau-1})=\Delta^+_{j^\tau}(\sigma^{\tau})\)\label{def:constr_kd_change}
    \item For all~\(\tau\in\{1,\dots,T\}\), for any~\(j\in\{1,\dots,n\}\) with~\(s_j(\sigma^{\tau-1})<s_j(\sigma)\), we must have that
    \[\kappa_{j^\tau}(\sigma^{\tau-1})- \Delta_{j^\tau}(\sigma^{\tau-1})\geq \kappa_{j}(\sigma^{\tau-1})- \Delta_{j^\tau}(\sigma^{\tau-1})\,.\] \label{def:constr_kd_order}
\end{enumerate}
By Condition \ref{def:constr_one_job}, we only speed up one job during each step. By Condition \ref{def:constr_kd_change}, we can only speed up this job until its affection or its shrinking energy changes. By Condition \ref{def:constr_kd_order}, in step~\(\tau\), the job~\(j^\tau\) we speed up must have maximum~\(\kappa_{j^\tau}(\sigma^{\tau-1})- \Delta_{j^\tau}(\sigma^{\tau-1})\) among all jobs that do not yet run at their speed in~\(\sigma\). Note that these conditions only restrict the order in which we speed up jobs and what is considered a step in the construction, and not the final speeds. Thus, such a construction must exist for any schedule~\(\sigma\).

\begin{observation}\label{obs:KD_decreases}
    Consider a schedule~\(\sigma\) with construction~\((\sigma^0,\dots,\sigma^T)\).
    Let~\(\tau>\tau'\), then for all~\(j\in\{1,\dots,n\}\), we have~\(\kappa_j(\sigma^\tau)\leq\kappa_j(\sigma^{\tau'})\),~\(\kappa^+_j(\sigma^\tau)\leq\kappa^+_j(\sigma^{\tau'})\),~\(\Delta_j(\sigma^\tau)\geq\Delta_j(\sigma^{\tau'})\), and~\(\Delta^+_j(\sigma^\tau)\geq\Delta^+_j(\sigma^{\tau'})\).
\end{observation}
\begin{proof}
    When we speed up a job, the completion times of all jobs can only decrease. Hence, their (lower) affection can only decrease. Furthermore, the shrinking (or expanding) energy of the sped-up job can only increase, while the shrinking energy of all other jobs remains unchanged.
\end{proof}

\begin{observation}\label{obs:no_affect_other_chain}
    If~\(j\notin \K_{j^\tau}(\sigma^{\tau-1})\), then~\(\K_j(\sigma^\tau)=\K_j(\sigma^{\tau-1})\). 
\end{observation}
\begin{proof}
    First, suppose~\(j\notin \K_{j^\tau}(\sigma^{\tau-1})\). Then~\(\K_{j^\tau}(\sigma^{\tau-1})\)and~\(\K_j(\sigma^{\tau-1})\) are completely disjoint. By the definition of affection, decreasing the processing time of~\(j^\tau\) can only affect the completion time of jobs in~\(K_{j^\tau}(\sigma^{\tau-1})\). Hence, for any~\(j'\in\K_j(\sigma^{\tau-1})\), we have that~\(C_{j'}(\sigma^{\tau-1})= C_{j'}(\sigma^{\tau})\). If follows that~\(\K_j(\sigma^\tau)=\K_j(\sigma^{\tau-1})\). 
\end{proof}

Since Algorithm \ref{alg:kappa_delta} always speeds up a job with maximum~\(\kappa-\Delta\), its steps also constitute a construction. Generally, in a construction, we are allowed to speed up a job which does not have maximum~\(\kappa-\Delta\), but only if all jobs with maximum~\(\kappa-\Delta\) already run at their final speed. If we do speed up such a job, then the construction can be the same as the construction from the Algorithm. Hence, we are interested in comparing constructions that do speed up a job with non-maximum~\(\kappa-\Delta\), to constructions that never speed up such a job. For this purpose, we introduce the following:

\begin{definition}
    Let $\sigma$ be a schedule with construction~\((\sigma^0,\dots,\sigma^T)\). A step~\(\tau\) follows the~\emph{\((\kappa-\Delta)\)-rule} if for all~\(j\in\{1,\dots,n\}\) we have
    \[
    \kappa_{j^\tau}(\sigma^{\tau-1})-\Delta_{j^\tau}(\sigma^{\tau-1})\geq \kappa_{j}(\sigma^{\tau-1})-\Delta_{j}(\sigma^{\tau-1})\,.
    \]
    If all steps follow the~\((\kappa-\Delta)\)-rule, then we say the construction follows the~\((\kappa-\Delta)\)-rule.
\end{definition}

By Corollary \ref{cor:kappa_delta_balance}, a construction of an optimal schedule cannot end after step~\(\tau\) if there is a job~\(j\) with~\(\kappa_j(\sigma^\tau)-\Delta_j(\sigma^\tau)>0\). However, it is not clear whether we should speed up jobs with~\(\kappa_j-\Delta_j=0\). To better compare different constructions, we may assume the following:

\begin{observation}\label{obs:KD_balance_contruction}
    Let~\(\sigma\) be an optimal schedule. Then w.l.o.g.\ we have that~\(\kappa_j(\sigma)-\Delta_j(\sigma)\leq 0\) and~\(\kappa^+_j(\sigma)-\Delta^+_j(\sigma)> 0\) for all~\(j\in\{1,\dots,n\}\).
\end{observation}
\begin{proof}
    By Corollary \ref{cor:kappa_delta_balance}, we know that for an optimal schedule~\(\sigma\), we have~\(\kappa_j(\sigma)-\Delta_j(\sigma)\leq 0\) and~\(\kappa^+_j(\sigma)-\Delta^+_j(\sigma)\geq 0\) for all~\(j\in\{1,\dots,n\}\). Suppose~\(\kappa^+_j(\sigma)-\Delta^+_j(\sigma)= 0\) for some job~\(j\). Then by Observation \ref{obs:affection} and Lemma \ref{lemma:delta}, there is an~\(\varepsilon>0\) such that~\(\Obj(\sigma)= \Obj(\sigma_j^{+\varepsilon})\). Furthermore, if we choose~\(\varepsilon\) large enough, then~\(\kappa^+_j(\sigma_j^{+\varepsilon})-\Delta^+_j(\sigma_j^{+\varepsilon})> 0\). Hence,~\(\sigma^{+\varepsilon}_j\) yields a new optimal schedule, where~\(\kappa^+_j(\sigma_j^{+\varepsilon})-\Delta^+_j(\sigma_j^{+\varepsilon})> 0\).
\end{proof}

To prove the correctness of Algorithm \ref{alg:kappa_delta}, we first prove the following: If we have a schedule with a construction that does not follow the~\((\kappa-\Delta)\)-rule, then we can find two jobs where slowing down one job and speeding up the other leads to an improvement in the objective value. Thus, such a schedule can never be optimal. To be more specific: the job we speed up will be one that would have adhered to the~\((\kappa-\Delta)\)-rule in a construction step where a violation occurred. The main challenge lies in identifying the job that should be slowed down in this exchange. Hence, we introduce the following lemma.

\begin{lemma}\label{lemma:exchange_last_vio}
    Let~\(\sigma\) be a schedule with construction~\((\sigma^0,\dots,\sigma^T)\), and~\(\kappa_j(\sigma)-\Delta_j(\sigma)\leq 0\) and~\(\kappa^+_j(\sigma)-\Delta^+_j(\sigma)>0\) for all~\(j\in\{1,\dots,n\}\). Suppose for some job~\(j\) and step~\(\tau\), we have that \[
    \kappa_j(\sigma^{\tau-1})- \Delta_j(\sigma^{\tau-1}) \geq \kappa_{j'}(\sigma^{\tau-1})- \Delta_{j'}(\sigma^{\tau-1})\,,
    \]
    for all~\(j'\in \K_j(\sigma^{\tau-1})\).
    Furthermore, suppose
   \(\kappa_j(\sigma^{\tau-1})-\Delta_j(\sigma^{\tau-1})> 0\),~\(\kappa_j(\sigma^{\tau})-\Delta_j(\sigma^{\tau})\leq 0\), and~\(j^\tau\neq j\). If all steps after~\(\tau\) follow the~\((\kappa-\Delta)\)-rule, then there is an~\(\varepsilon>0\) such that
    \[
    \Obj(\sigma^\varepsilon_{j^\tau j}) \leq \Obj(\sigma) - \varepsilon\left(\left(\kappa_{j}(\sigma^{\tau-1})-\Delta_{j}(\sigma^{\tau-1})\right) - \left(\kappa_{j^\tau}(\sigma)-\Delta_{j^\tau}(\sigma^{\tau-1})\right)\right)\,.
    \]
\end{lemma}
\begin{proof}
    By Observation \ref{obs:no_affect_other_chain}, we must have~\(\K_j(\sigma^{\tau-1})=\K_{j^\tau}(\sigma^{\tau-1})\). 
    We distinguish between two cases
    \proofsubparagraph*{Case 1:~\(j<j^\tau\)}
    In this case, we have that~\(\K_j(\sigma^{\tau})=\K_{j^\tau}(\sigma^{\tau})\). Since 
    \[\kappa_j(\sigma^{\tau-1})- \Delta_j(\sigma^{\tau-1}) \geq \kappa_{j'}(\sigma^{\tau-1})- \Delta_{j'}(\sigma^{\tau-1})\] 
    for all~\(j'\in\K_j(\sigma^{\tau-1})\) and by Observation \ref{obs:overlap_chain}, we have that~\(\kappa_{j'}(\sigma^{\tau-1})- \Delta_{j'}(\sigma^{\tau-1})\leq0\) for all~\(j'\in \K_j(\sigma^\tau)\). Since~\(\kappa^+_j(\sigma)-\Delta^+_j(\sigma)> 0\) for all~\(j\in\{1,\dots,n\}\), it follows that none of the jobs in~\(\K_j(\sigma^\tau)\) will be sped up in any of the steps after~\(\tau\). Thus,~\(\K_j(\sigma)= \K_{j^\tau}(\sigma)\) and by Observation \ref{obs:affection_from_chain}
    \[
    \kappa^+_j(\sigma)-\kappa^+_{j^\tau}(\sigma) = \kappa^+_j(\sigma^{\tau-1})-\kappa^+_{j^\tau}(\sigma^{\tau-1})\,.
    \]
    By the definition of a construction,~\(\Delta^+_{j^\tau}(\sigma) = \Delta^+_{j^\tau}(\sigma^\tau) = \Delta_{j^\tau}(\sigma^{\tau-1}) \). The result follows from Lemma \ref{lemma:exchange_cases}.

    \proofsubparagraph*{Case 2:~\(j>j^\tau\)}
    By the definition of a construction, we have that~\(\K^+_j(\sigma^\tau)= \K^+_{j^\tau}(\sigma^\tau)\). Furthermore, by Observation \ref{obs:earlier_loses_more}, we have that for all~\(j'\in \K_j(\sigma^{\tau-1})\) with~\(j'<j\)
    \[
    \kappa_{j'}(\sigma^{\tau-1})-  \kappa_{j'}(\sigma^\tau) \geq \kappa_{j'}(\sigma^{\tau-1})-  \kappa_{j'}(\sigma^\tau)\,.
    \]
    Since~\(\kappa_j(\sigma^{\tau-1})-\Delta_j(\sigma^{\tau-1})\geq \kappa_{j'}(\sigma^{\tau-1})-\Delta_{j'}(\sigma^{\tau-1})\), it follows that 
    \[
    \kappa_{j'}(\sigma^{\tau})-\Delta_{j'}(\sigma^{\tau})\leq \kappa_{j}(\sigma^{\tau})-\Delta_{j}(\sigma^{\tau})\leq0\,.
    \]
    Hence, none of these jobs~\(j'\) will be sped up in any of the steps after~\(\tau\), and thus we have~\(\K^+_j(\sigma)=\K^+_{j^\tau}(\sigma)\). By Observation \ref{obs:overlap_chain}, it follows that
    \[
    \kappa^+_j(\sigma)-\kappa^+_{j^\tau}(\sigma) = \kappa^+_j(\sigma^{\tau-1})-\kappa^+_{j^\tau}(\sigma^{\tau-1})\,.
    \]
    By the definition of a construction,~\(\Delta^+_{j^\tau}(\sigma) = \Delta^+_{j^\tau}(\sigma^\tau) = \Delta_{j^\tau}(\sigma^{\tau-1}) \). The result follows from Lemma \ref{lemma:exchange_cases}.
\end{proof}

\begin{corollary}\label{cor:must_follow_KD}
    Any construction of an optimal schedule follows the~\((\kappa-\Delta)\)-rule.
\end{corollary}
\begin{proof}
    Suppose there is an optimal schedule~\(\sigma\) with a construction~\((\sigma^0,\dots,\sigma^T)\) that does not follow the~\((\kappa-\Delta)\)-rule. Let~\(\tau'\) be the last violation of the rule, where for some job~\(j\), we have 
    \[
    \kappa_j(\sigma^{\tau'-1})-\Delta_j(\sigma^{\tau'-1})\geq \kappa_{j'}(\sigma^{\tau'-1})-\Delta_{j'}(\sigma^{\tau'-1})\,,
    \]
    for all~\(j'\in\{1,\dots,n\}\) and
    \[
    \kappa_j(\sigma^{\tau'-1})-\Delta_j(\sigma^{\tau'-1})>\kappa_{j^{\tau'}}(\sigma^{\tau'-1})-\Delta_{j^{\tau'}}(\sigma^{\tau'-1})> 0\,,
    \]
    where the last inequality is a necessary condition to speed up~\(j^{\tau'}\). Now let~\(\tau\geq\tau'\) be such that~\(\kappa_j(\sigma^{\tau-1})-\Delta_j(\sigma^{\tau-1})\geq 0\) and~\(\kappa_j(\sigma^{\tau})-\Delta_j(\sigma^{\tau})< 0\). By the definition of a construction, we must have 
    \[
    \kappa_j(\sigma^{\tau'-1})-\Delta_j(\sigma^{\tau'-1})>\kappa_{j^\tau}(\sigma^{\tau-1})-\Delta_{j^\tau}(\sigma^{\tau'-1})
    \]
    otherwise, we would have to speed up job~\(j^\tau\) in step~\(\tau'\) and prevent a violation. From Observation \ref{obs:no_affect_other_chain}, it follows that~\(\K_j(\sigma^{\tau-1})= \K_{j^\tau}(\sigma^{\tau-1})\), and thus by Observation
    \ref{obs:overlap_chain}, we have
    \[
    \kappa_j(\sigma^{\tau-1})-\Delta_j(\sigma^{\tau-1})>\kappa_{j^\tau}(\sigma^{\tau-1})-\Delta_{j^\tau}(\sigma^{\tau-1})\,.
    \]
    From Lemma \ref{lemma:exchange_last_vio}, we obtain that there is an~\(\varepsilon>0\) such that~\(\Obj(\sigma_{j^\tau j})<\Obj(\sigma)\), which contradicts that~\(\sigma\) is optimal.
\end{proof}

While Corollary \ref{cor:must_follow_KD} shows it is necessary for a construction of an optimal schedule to follow the~\((\kappa-\Delta)\)-rule, it remains to show that it is sufficient (if the schedule also satisfies the conditions from Observation \ref{obs:KD_balance_contruction})). Thus, we cannot yet conclude that the schedule constructed by Algorithm \ref{alg:kappa_delta} is indeed optimal. To complete our proof, we show that any schedule with a construction that follows the~\((\kappa-\Delta)\)-rule that also satisfies the conditions from Observation \ref{obs:KD_balance_contruction} is optimal. To do this, we use Lemma \ref{lemma:exchange_last_vio} to show that for two schedules that both satisfy these conditions, we can make one of them ``more similar" to the other, without changing its objective value. Hence, their objective value must be the same. To describe more concretely what ``more similar" means, we introduce the following notation.

For a schedule~\(\sigma\) with construction~\((\sigma^0,\dots,\sigma^T)\), we denote by~\((\sigma^0,\dots,\sigma^T)_Y\) the first part of the construction until we have shrunk the total processing time by~\(Y\). To be more precise: we have that~\((\sigma^0,\dots,\sigma^T)_Y= (\sigma^0,\dots,\sigma^{\tau-1}, \Tilde{\sigma}^\tau)\) where
\[
 \sum_{j=1}^n x_j(\sigma^{\tau-1})\geq \sum_{j=1}^n x_j(\sigma^{0})-Y > \sum_{j=1}^n x_j(\sigma^{\tau})
\]
and 
\[
x_{j^\tau}(\Tilde{\sigma}^\tau) =  \sum_{j=1}^n x_j(\sigma^{\tau-1})- \left(\sum_{j=1}^n x_j(\sigma^{0}) - Y \right)\,.
\]
When two schedules~\(\sigma\) and~\(\Tilde{\sigma}\) have constructions~\((\sigma^0,\dots,\sigma^T)\) and~\((\Tilde\sigma^0,\dots,\Tilde\sigma^{\Tilde{T}})\) with 
\[
(\Tilde\sigma^0,\dots,\Tilde\sigma^{\Tilde{T}})_Y= (\sigma^0,\dots,\sigma^T)\,
\] 
then we say the constructions agree for~\(Y\).

\begin{lemma}\label{lemma:KD_opt}
    Let~\(\sigma\) be a schedule with~\(\kappa_j(\sigma)-\Delta_j(\sigma)\leq 0\) and~\(\kappa^+_j(\sigma)-\Delta^+_j(\sigma)>0\), and a construction~\((\sigma^0,\dots,\sigma^T)\), that follows the~\((\kappa-\Delta)\)-rule. Then~\(\sigma\) is optimal.
\end{lemma}
\begin{proof}
    Suppose there is a schedule~\(\sigma\) with~\(\kappa_j(\sigma)-\Delta_j(\sigma)\leq 0\) and~\(\kappa^+_j(\sigma)-\Delta^+_j(\sigma)>0\), and a construction~\((\sigma^0,\dots,\sigma^T)\), that follows the~\((\kappa-\Delta)\)-rule but is not optimal. Let~\(\Tilde\sigma\) be an optimal schedule with construction~\((\Tilde\sigma^0,\dots,\Tilde\sigma^{\Tilde{T}})\) agrees with the construction of~\(\sigma\) for~\(Y\). Furthermore, for the sake of contradiction, let there be no optimal schedule with a construction that agrees with the construction of~\(\sigma\) for~\(Y'\) with~\(Y'>Y\). 
    
    Note that if the constructions start deviating from each other during one of the steps, then we can ``split" this step into two and maintain a valid construction. Hence, we may assume that the two constructions agree for exactly~\(\tau'-1\) steps. Let~\(j\) the job sped up in~\(\sigma\) in step~\(\tau'\). Then we must have that~\(\kappa_j(\sigma^{\tau'-1})-\Delta_j(\sigma^{\tau'-1})\geq 0\). Hence, there must be some step~\(\tau\geq \tau'\) such that~\(\kappa_j(\sigma^{\tau-1})-\Delta_j(\sigma^{\tau-1})> 0\) and~\(\kappa_j(\sigma^{\tau})-\Delta_j(\sigma^{\tau})\leq 0\). Let~\(j^\tau\) be the job sped up in in step~\(\tau\) of the construction of~\(\Tilde{\sigma}\). Then by Observation \ref{obs:no_affect_other_chain}, we have~\(\K_j(\Tilde\sigma^{\tau-1})=\K_{j^\tau}(\Tilde\sigma^{\tau-1})\), and it follows from Observation \ref{obs:overlap_chain} that 
    \[
    \kappa_j(\Tilde{\sigma}^{\tau'-1})- \kappa_j(\Tilde{\sigma}^{\tau-1})  = \kappa_{j^\tau}(\Tilde{\sigma}^{\tau'-1})- \kappa_{j^\tau}(\Tilde{\sigma}^{\tau-1})\,.
    \]
    Note that~\(j\) cannot be sped up in any of the steps starting from~\(\tau'\), since otherwise we could have a construction that agrees with the construction of~\(\sigma\) for more than~\(Y\). Hence, we have~\(\Delta_j(\Tilde\sigma^{\tau-1})=\Delta_j(\Tilde\sigma^{\tau -1})\). Since~\(\Tilde{\sigma}\) is optimal, it follows from Corollary \ref{cor:must_follow_KD} that~\(\Tilde{\sigma}\) follows the~\((\kappa-\Delta)\)-rule. Thus, we have that
    \begin{align*}
        \kappa_{j^{\tau}}(\Tilde{\sigma}^{\tau-1})- \Delta_{j^{\tau}}(\Tilde{\sigma}^{\tau-1})&\geq \kappa_j(\Tilde{\sigma}^{\tau-1})- \Delta_j(\Tilde{\sigma}^{\tau-1}) = \kappa_j(\Tilde{\sigma}^{\tau-1})- \Delta_j(\Tilde{\sigma}^{\tau'-1})\\ 
        &= \kappa_j(\Tilde{\sigma}^{\tau-1})- \kappa_j(\Tilde{\sigma}^{\tau'-1}) + \kappa_j(\Tilde{\sigma}^{\tau'-1})  -  \Delta_j(\Tilde{\sigma}^{\tau'-1}) \\
        &= \kappa_{j^\tau}(\Tilde{\sigma}^{\tau-1})- \kappa_{j^\tau}(\Tilde{\sigma}^{\tau'-1}) + \kappa_j(\Tilde{\sigma}^{\tau'-1})  -  \Delta_j(\Tilde{\sigma}^{\tau'-1}) \\
        &\geq \kappa_{j^\tau}(\Tilde{\sigma}^{\tau-1})- \kappa_{j^\tau}(\Tilde{\sigma}^{\tau'-1}) + \kappa_{j^\tau}(\Tilde{\sigma}^{\tau'-1})  -  \Delta_{j^\tau}(\Tilde{\sigma}^{\tau'-1}) \\
        &= \kappa_{j^\tau}(\Tilde{\sigma}^{\tau-1})  -  \Delta_{j^\tau}(\Tilde{\sigma}^{\tau'-1})\,.
    \end{align*}
    By Observation \ref{obs:KD_decreases}, we have~\(\Delta_{j^\tau}(\Tilde{\sigma}^{\tau'-1})\leq \Delta_{j^{\tau}}(\Tilde{\sigma}^{\tau-1})\), and thus we have that 
    \[
    \kappa_{j^{\tau}}(\Tilde{\sigma}^{\tau-1})- \Delta_{j^{\tau}}(\Tilde{\sigma}^{\tau-1})= \kappa_j(\Tilde{\sigma}^{\tau-1})- \Delta_j(\Tilde{\sigma}^{\tau-1})\,.
    \] 
    By Lemma \ref{lemma:exchange_last_vio}, we have that~\(\Obj(\Tilde\sigma_{j^\tau j}^\varepsilon)=\Obj(\Tilde\sigma)\). Thus~\(\Tilde\sigma_{j^\tau j}^\varepsilon\) is also an optimal schedule and has a construction that agrees for~\(\varepsilon\) more shrinkage in processing time with step~\(\tau\) of the construction of~\(\sigma\). This gives us an optimal schedule with a construction that agrees with the construction of~\(\sigma\) for some~\(Y'\) with~\(Y'>Y\), leading to a contradiction with our choice of~\(\Tilde{\sigma}\).
\end{proof}

Since Algorithm \ref{alg:kappa_delta} always follows the~\((\kappa-\Delta)\)-rule and yields a schedule~\(\sigma\) that satisfies~\(\kappa_j(\sigma)-\Delta_j(\sigma)\leq 0\) and~\(\kappa^+_j(\sigma)-\Delta^+_j(\sigma)>0\) for all~\(j\in\{1,\dots,n\}\), the correctness of Algorithm \ref{alg:kappa_delta} follows directly from Lemma \ref{lemma:KD_opt}.

 Finally, let us consider the running time of Algorithm \ref{alg:kappa_delta}. Since the steps of the algorithm constitute a construction, it follows from Observation~\ref{obs:KD_decreases} for any job~\(j\) we have that~\(\kappa_{j}\) never increases and~\(\Delta_{j}\) never decreases. Thus, a step where we speed up job~\(j\) ends when either~\(\kappa_j\) or~\(\Delta_j\) decreases. Since the affection of any job~\(j\) is upper bounded by~\(n\), it follows that~\(\kappa_j\) cannot decrease more than~\(n\). Similarly,~\(\Delta_j\) cannot increase more than~\(k\). It follows that the number of steps is~\(O(n^3k)\). Since we can also compute in polynomial time how much we can increase the speed of a job until it either reaches the next fixed speed or loses affection, it follows that Algorithm \ref{alg:kappa_delta} runs in polynomial time, proving \textbf{Theorem \ref{thm:kappa_delta}}. 
\end{document}